\documentclass[journal,onecolumn,draftcls]{IEEEtran}

\ifCLASSINFOpdf
\else
\fi
\usepackage{amsmath,amssymb,amsfonts}
\usepackage{algorithmic}
\usepackage{graphicx}
\usepackage{textcomp}
\usepackage{xcolor}
\usepackage{tabularx,booktabs}
\usepackage{amsthm}
\usepackage{bbm}
\usepackage{csquotes}
\usepackage{chngcntr}
\usepackage{apptools}
\usepackage{tikz}
\usepackage{caption}  
\usetikzlibrary{mindmap,backgrounds,arrows.meta}
\newcommand{\A}{\mathcal{A}}
\newcommand{\B}{\mathcal{B}}
\newcommand{\X}{\mathcal{X}}
\newcommand{\Y}{\mathcal{Y}}
\newcommand{\Z}{\mathcal{Z}}
\newcommand{\Le}{\mathcal{L}}
\newcommand{\E}{\mathbb{E}}
\newcommand{\Prob}{\mathbb{P}}
\newcommand{\Pm}{\mathcal{P}}
\newcommand{\Q}{\mathcal{Q}}
\newcommand{\F}{\mathcal{F}}

\DeclareMathOperator*{\esssup}{ess\,sup}

\theoremstyle{plain}
\newtheorem{Theorem}{Theorem}
\newtheorem{Lemma}{Lemma}
\newtheorem{Corollary}{Corollary}

\newtheorem*{Corollary*}{Corollary}
\newtheorem{Proposition}{Proposition}
\newtheorem{Example}{Example}

\theoremstyle{definition}
\newtheorem{Definition}{Definition}

\theoremstyle{remark}

\newtheorem{Remark}{Remark}
\makeatletter
\newcounter{labelcnt}
\renewcommand{\thelabelcnt}{(\alph{labelcnt})}
\newcommand{\setlabel}[1]{%
  \refstepcounter{labelcnt}\ltx@label{lbl:#1}%
  {\text{\upshape\thelabelcnt}}%
}
\newcommand{\reflabel}[1]{\text{\upshape\ref{lbl:#1}}}
\makeatother

\newcommand{\ml}[2]{\mathcal{L}\left(#1  \!\!  \to  \!\!   #2\right)} 

\newcommand{\argmax}{\operatornamewithlimits{argmax}}

\begin{document}
%

\title{Generalization Error Bounds Via R\'enyi-, $f$-Divergences and Maximal Leakage}
%
%
%

\author{Amedeo~Roberto~Esposito,~\IEEEmembership{Student~Member,~IEEE,}
        Michael~Gastpar,~\IEEEmembership{Fellow,~IEEE,}
        and~Ibrahim~Issa,~\IEEEmembership{Member,~IEEE}
\thanks{The work in this manuscript was supported in part by the Swiss National Science Foundation under Grant 169294, and by EPFL. The work in this manuscript was partially presented at the {\it 2019 IEEE International Symposium on Information Theory,} Paris, France, and at the {\it 2019 IEEE Information Theory Workshop,} Visby, Sweden, the {\it 2020 International Zürich Seminar} and the {\it 2020 IEEE International Symposium on Information Theory}.}
\thanks{A. Esposito and M. Gastpar are with the School of Computer and Communication Sciences, {\'E}cole Polytechnique F{\'e}d{\'e}rale de Lausanne (EPFL), Lausanne, Switzerland, e-mail: \{amedeo.esposito, michael.gastpar\}@epfl.ch}
\thanks{I. Issa is with the American University of Beirut.}
}

\maketitle

\begin{abstract}
In this work, the probability of an event under some joint distribution is bounded by measuring it with the product of the marginals instead (which is typically easier to analyze) together with a measure of the dependence between the two random variables.
These results find applications in adaptive data analysis, where multiple dependencies are introduced and in learning theory, where they can be employed to bound the generalization error of a learning algorithm.
Bounds are given in terms of Sibson's Mutual Information, $\alpha-$Divergences, Hellinger Divergences, and $f-$Divergences. 
A case of particular interest is the Maximal Leakage (or Sibson's Mutual Information of order infinity), since this measure is robust to post-processing and composes adaptively.
The corresponding bound can be seen as a generalization of classical bounds, such as Hoeffding's and McDiarmid's inequalities, to the case of dependent random variables.
\end{abstract}


%
\IEEEpeerreviewmaketitle

\section{Introduction}

Let us consider two probability spaces $(\Omega,\F,\Pm),(\Omega,\F,\Q)$ and let $E\in \F$ be a measurable event. 
Our aim is to provide bounds of the following form: 
\begin{equation}\Pm(E) \leq f(\Q(E))\cdot g(d\Pm/d\Q),\label{generalBound}\end{equation} for some functions $f,g$. $E$ represents some \enquote{undesirable} event
(e.g., large generalization error), whose measure under $\Q$ is known and whose measure under $\Pm$ we wish to bound. $d\Pm/d\Q$ denotes the Radon-Nikodym derivative of $\Pm$ with respect to $\Q$ (assuming it exists).  $g(d\Pm/d\Q)$ is often going to be a function of some divergence between $\Pm$ and $\Q$, $h(D(\Pm\|\Q))$ (e.g., KL, R\'enyi's $\alpha-$Divergence, etc.). 
Of particular interest is the case where $\Omega = \X\times\Y$,  $\Pm = \Pm_{XY}$ (the joint distribution), and $\Q=\Pm_X\Pm_Y$ (product of the marginals). This allows us to bound the likelihood of $E \subseteq \X \times \Y$ when two random variables $X$ and $Y$ are dependent as a function of the likelihood of $E$ when $X$ and $Y$ are independent (a scenario typically much easier to analyze). Such a result can be applied in the analysis of the generalization error of learning algorithms, as well as in adaptive data analysis (with a proper choice of the dependence measure). 
Adaptive data analysis is a recent field that is gaining attention due to its connection with the \enquote{Reproducibility Crisis} \cite{genAdap,maxInfo}. 
The idea is that, whenever you apply a sequence of analyses to some data (\textit{e.g.}, data-exploration procedures) and each analysis informs the subsequent ones, even though each of these algorithms is guaranteed to generalize well in \textit{isolation}, this may no longer be true when they are \textit{composed} together. The problem that arises with the composition is believed to be connected with the \textit{leakage} of information from the data. The leakage happens because the output of each algorithm becomes an input to the subsequent ones.
In order to be used in adaptive data analysis, a measure that provides such bounds needs to be robust to post-processing and to compose adaptively (meaning that we can bound the measure between input and output of the composition of the sequence of algorithms if each of them has bounded measure). Results of this form involving mutual information can be found in \cite{learningMI,infoThGenAn,explBiasMI}.
Via inequalities like in \eqref{generalBound} we can provide bounds for adaptive mechanisms by treating them as non-adaptive and paying a \enquote{penalty} (\textit{e.g.}, a measure of statistical dependency) that estimates how far is the mechanism from being non-adaptive. 
With this aim, we first provide general bounds in terms of Luxemburg norms, Amemiya norms, and $f$-mutual information.
As corollaries, we derive several families of interesting bounds in the form of~\eqref{generalBound} with $\Pm=\Pm_{XY}$ and $\Q=\Pm_X \Pm_Y$:
\begin{itemize}
\item a family of bounds involving Sibson's Mutual Information of order $\alpha$;
\item a bound involving Maximal Leakage \cite{leakageLong};
\item a family of bounds involving the R\'enyi's and Hellinger divergences of order $\alpha$;
\item a bound involving Hellinger squared distance.
\end{itemize}
A representation of our results and their connections is given in Figure~\ref{mindmap} below.
We focus in particular on the bounds involving Maximal Leakage, which is a secrecy metric that has appeared both in the computer security literature~\cite{CompSecQuantLeakage}, and the information theory literature~\cite{leakage}. It quantifies the leakage of information from a random variable $X$ to another random variable $Y$, and is denoted by $\ml{X}{Y}$. The basic insight is as follows: if a learning algorithm leaks little information about the training data, then it will generalize well. Moreover, similarly to differential privacy, maximal leakage behaves well under composition: we can bound the leakage of a sequence of algorithms if each of them has bounded leakage. It is also robust under post-processing.
In addition, the expression to compute it is simply given by the following formula (for finite $X$ and $Y$): 
\begin{equation} \label{eq:ml1}
\ml{X}{Y}= \log \sum_y \max_{x: P(x)>0} P_{Y|X}(y|x),
\end{equation}
making it more amenable to analysis and relatively easy to compute, especially for algorithms whose randomness consists in adding independent noise to the outcomes. 
Despite the main focus being on a joint distribution and the corresponding product of the marginals, the proof techniques are more general and can be applied to any pair of joint distributions (under a mild condition of absolute continuity). Moreover, the Maximal Leakage result, as well as the bound using infinite-R\'enyi divergence, reduce to the classical concentration inequalities when independence holds (i.e., $\Pm_{XY}=\Pm_X \Pm_Y$).
\subsection{Further related work}
In addition to differentially private algorithms, Dwork \emph{et al.}~\cite{genAdap} show that algorithms whose output can be described concisely generalize well. They further introduce $\beta$-max information to unify the analysis of both classes of algorithms. Consequently, one can provide generalization guarantees for a sequence of algorithms that alternate between differential privacy and short description. In~\cite{maxInfo}, the authors connect $\beta$-max information with the notion of approximate differential privacy, but show that there are no generalization guarantees for an arbitrary composition of algorithms that are approximate-DP and algorithms with short description length. With a more information-theoretic approach, bounds on the exploration bias and/or the generalization error are given in~\cite{explBiasMI,infoThGenAn,jiao2017dependence,ISIT2018, chainingMI, tighteningMI,genErrMISGLD}, using mutual information and other dependence-measures. Some results have also been found using Wasserstein distance \cite{genErrWassDist,genErrWassDist2}. 

\subsection{Notation}
We will denote by calligraphic letters $\Pm, \Q$ probability measures and with capital letters $X,Y,Z$ random variables. Given two measures $\Pm,\Q$, $\Pm\ll\Q$ denotes the concept of absolute continuity, \textit{i.e.}, for any measurable set $E$, $\Q(E)=0\implies \Pm(E)=0$.
Given two random variables $X,Y$ over the spaces $\mathcal{X},\mathcal{Y}$ we will denote by $\Pm_{XY}$ a joint measure over the product space $\mathcal{X}\times\mathcal{Y}$, while with $\Pm_X\Pm_Y$ we will denote the product of the marginals, \textit{i.e.}, for any measurable set $E\subseteq \mathcal{X}\times\mathcal{Y}, \Pm_X\Pm_Y(E) = \int_{(x,y)\in E} d\Pm_X(x)d\Pm_Y(y)$.\\
Given a probability measure $\Pm$ and a random variable $X$ defined over the same space, we will denote with \begin{equation}\mathbb{E}_\Pm[X] = \int x d\Pm(x).\end{equation}\\
Furthermore, given a zero-mean random variable $X$ we say that it is $\sigma^2$-sub-Gaussian if the following holds true for every $\lambda\in\mathbb{R}$:
\begin{equation}
\mathbb{E}[e^{\lambda X}] \leq e^{\frac{\lambda^2\sigma^2}{2}}.
\end{equation}
For the remainder of this paper $\log$ is always taken to the
base $e$.
\subsection{Overview}
In Section \ref{background} we define the fundamental objects that will be used in this work: 
\begin{itemize}
    \item In Subsection \ref{infoMeas} we consider R\'enyi's-$\alpha$ Divergences, Sibson's Mutual Information, Maximal Leakage and $f-$divergences;
    \item In Subsection \ref{learning} we provide an overview of the basic concepts in Learning Theory;
\end{itemize} 
In Section \ref{mainResults} we prove our most general results:
\begin{itemize}
    \item Theorem \ref{orliczAmemiyaBound} and \ref{luxemburgNormBound} are the most general one and involves Luxemburg and Amemiya norms;
    \item Theorem \ref{fDivBound} bounds $\Pm_{XY}(E)$ with a function of $I_\phi(\Pm_{XY}\|\Pm_X\Pm_Y)$ and $\Pm_X\Pm_Y(E)$;
    \item Theorem \ref{alphaExpBound} bounds $\Pm_{XY}(E)$ using norms of $\Pm_X(E_Y)$ and the Radon-Nikodym derivative $d\Pm_{XY}/d\Pm_X\Pm_Y$.
\end{itemize}
Then we specialize these results and obtain bounds involving:
\begin{itemize}
    \item Sibson's $\alpha-$Mutual Information (Section \ref{mainResults});
    \item Maximal Leakage (Section \ref{sec:SMI});
    \item Hellinger and $\alpha$-Divergences (Section \ref{sec:ML});
\end{itemize}
In each of these section we also show how to apply these results to bound the generalization error.
In Section \ref{adaptiveDA} we consider the basic definitions of Adaptive Data Analysis and show how some of our results can be employed in the area. 
To conclude, in Section \ref{comparison} we compare our results with recent results in the literature. Some extension of our bounds to expected generalization error is also considered in Appendix \ref{sec:expectedGenErr}.

\begin{tikzpicture}[mindmap]\label{mindmap}


  \begin{scope}[mindmap, concept color=orange!70, text=white]
    \node [concept] (ame) {Theorem \ref{orliczAmemiyaBound} \\ Luxemburg \& Amemiya Norms} 
      child {node [concept] (hold) {Theorem \ref{alphaExpBound} \\ H\"older's Conjugates}
      child [minimum size = 70] {node [concept] (alpha) at (-1,0) {Corollary \ref{alphaDivBound} \\ $\alpha-$Divergences}}
      child [grow=160] {node [concept] (sibs) at (-1,-0.5) {Corollary \ref{sibsMIBoundCor} \\ Sibson's $\alpha$-Mutual Information}
      child [minimum size=50, grow=175, font= \fontsize{8pt}{8pt}\selectfont] {node [concept] (max) at (-0.8,0) {Corollary \ref{adaptML} \\ Maximal Leakage}}
      }
      };
  \end{scope}


  \begin{scope}[mindmap, concept color=red!40,text=white]
    \node [concept] (lux) at (-6,-10) {Theorem \ref{luxemburgNormBound}\\ Luxemburg Norms}; 
  \end{scope}


  \begin{scope}[mindmap, concept color=teal!40,text=white]
    \node [concept] (fdiv) at (6.5,-10) {Theorem \ref{fDivBound} \\ $\phi-$Divergences} 
      child [grow=165, level distance=120] {node [concept] (helDiv) at (-0.3,0) {Corollary \ref{generalHellingerBound} \\ Hellinger $p$-Divergences}
      child[grow = 220] {node [concept] (chi) at (-0.4,0.1) {Corollary \ref{generalHellingerBound} \\ $\chi^2$-Divergence}}
      }
      child [grow = 100] {node [concept] (helSq) at (-2,-1.3) {Corollary \ref{HellBound} \\ Hellinger Square Distance} };
  \end{scope}



  \begin{scope}[on background layer]
    \foreach \i/\j/\k/\l in
    {%
      lux/red!40/alpha/orange!70,
      helDiv/teal!40/alpha/orange!70
    }
    \path (\i) to[circle connection bar switch color=from (\j) to (\l)] (\k);
  \end{scope}
   \begin{pgfonlayer}{background}    
    \draw 
      [concept connection,->,orange!70,shorten >= 0.5pt,shorten <= 20pt, -{Stealth[angle=70:1pt 6]}] 
      (lux) to (alpha);
       \draw 
      [concept connection,->,orange!70,shorten >= -0.10pt,-{Stealth[angle=70:1pt 6]}] 
      (ame) to (hold);
       \draw 
      [concept connection,->,orange!70,shorten >= -0.10pt,-{Stealth[angle=70:1pt 6]}] 
      (hold) to (sibs);
       \draw 
      [concept connection,->,orange!70,shorten >= -0.10pt,-{Stealth[angle=70:1pt 6]}] 
      (sibs) to (max);
      \draw 
      [concept connection,->,teal!40,shorten >= -0.5pt,-{Stealth[angle=70:1pt 6]}] 
      (fdiv) to (helSq);
      \draw 
      [concept connection,->,teal!40,shorten >= -0.5pt,-{Stealth[angle=70:1pt 6]}] 
      (fdiv) to (helDiv);
       \draw 
      [concept connection,->,teal!40,shorten >= -0.5pt,-{Stealth[angle=70:1pt 6]}] 
      (helDiv) to (chi);
      
      \draw 
      [concept connection,->,orange!70,shorten >= -0.5pt,-{Stealth[angle=70:1pt 6]}] 
      (helDiv) to (alpha);
        \draw 
      [concept connection,->,teal!40,shorten >= 1pt,shorten <= 20pt, -{Stealth[angle=70:1pt 6]}] 
      (alpha) to (helDiv);
    \end{pgfonlayer}
\end{tikzpicture}


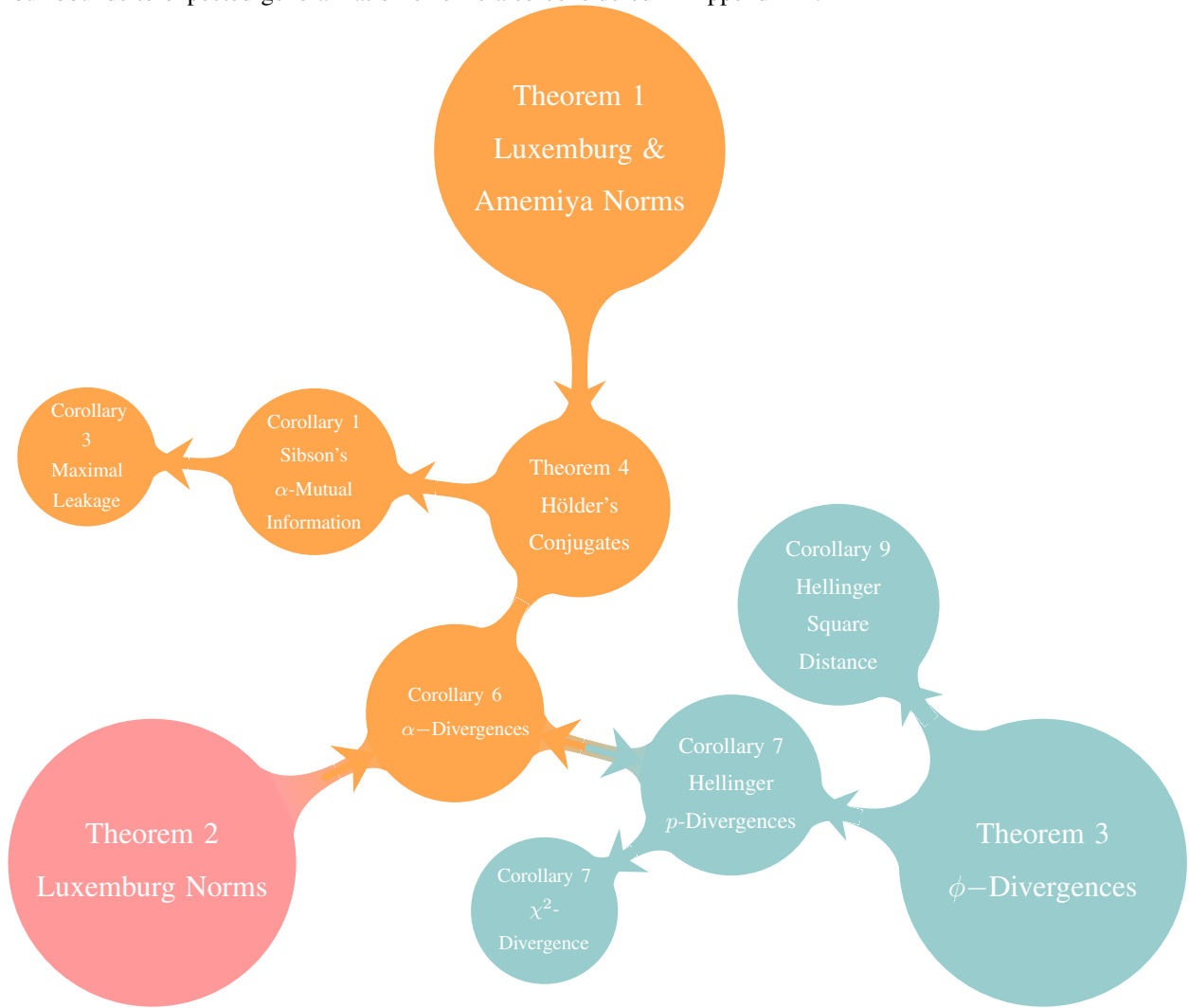
\captionof{figure}{A graphical representations of our main results and how they connect to each other.}

\section{Background And Definitions}\label{background}

\subsection{Information Measures}\label{infoMeas}
We will now briefly introduce the information measures that we will use to provide bounds. The idea is to try and capture the dependency between two random variables $X,Y$ through some information measure and employ it in order to provide bounds. We will consider $X$ to be the input of a learning algorithm $\A$ and $Y=\A(X)$ the corresponding (random) output. By controlling some measure of dependency, we will control how much the learning algorithm $\A$ is overfitting to the data. 
\subsubsection{R\'enyi's $\alpha-$Divergence}\label{renyisDiv}
Introduced by R\'enyi in an attempt to generalize the concept of Entropy and KL-Divergence, the $\alpha$-Divergence has then found many applications over the years in hypothesis testing, guessing and several other statistical inference problems \cite{verduAlpha}.
Indeed, it has several useful operation interpretations  (e.g., the number of bits by which a mixture of two codes can be compressed, the cut-off rate in block coding and hypothesis testing \cite{alphaDiv,opMeanRDiv1}\cite[p. 649]{opMeanRDiv2}).
It can be defined as follows \cite{alphaDiv}:
\begin{Definition}
	Let $(\Omega,\F,\Pm),(\Omega,\F,\Q)$ be two probability spaces. Let $\alpha>0$ be a positive real different from $1$. Consider a measure $\mu$ such that $\Pm\ll\mu$ and $\Q\ll\mu$ (such a measure always exists, e.g. $\mu=(\Pm+\Q)/2$)) and denote with $p,q$ the densities of $\Pm,\Q$ with respect to $\mu$. The $\alpha-$Divergence of $\Pm$ from $\Q$ is defined as follows:
	\begin{align}
	D_\alpha(\Pm\|\Q)=\frac{1}{\alpha-1} \log \int p^\alpha q^{1-\alpha} d\mu.
	\end{align}
\end{Definition}
\begin{Remark}
    The definition is independent of the chosen measure $\mu$ whenever $\infty>\alpha>0$ and $\alpha\neq 1$. It is indeed possible to show that
    $\int p^{\alpha}q^{1-\alpha} d\mu = \int \left(\frac{q}{p}\right)^{1-\alpha}d\Pm $, and that whenever $\Pm\ll\Q$ or $0<\alpha<1$ $\int p^{\alpha}q^{1-\alpha} d\mu= \int \left(\frac{p}{q}\right)^{\alpha}d\Q $, see \cite{alphaDiv}.
\end{Remark}
It can be shown that if $\alpha>1$ and $\Pm\not\ll\Q$ then $D_\alpha(\Pm\|\Q)=\infty$. The behavior of the measure for $\alpha\in\{0,1,\infty\}$ can be defined by continuity. In particular, we have that $\lim_{\alpha\to1}D_\alpha(\Pm\|\Q) = D(\Pm\|\Q)$, \textit{i.e.}, the classical Kullback-Leibler divergence. For an extensive treatment of $\alpha$-Divergences and their properties we refer the reader to \cite{alphaDiv}. 
\subsubsection{Sibson's $\alpha-$Mutual Information}\label{sibson}
Starting from the notion of \enquote{information radius}, Sibson built a generalization of mutual information that retains many interesting properties~\cite{infoRadius}. Although defined in a different way Sibson's $\alpha$-Mutual Information $I_\alpha(X,Y)$\footnote{Throughout the work we will denote with a comma \textbf{asymmetric} information measures (like $I_\alpha(X,Y)$) and with a semicolon \textbf{symmetric} information measures (like $I(X;Y)$).} can be re-defined in terms of $\alpha-$Divergences\cite{verduAlpha}:
\begin{Definition}
    Let $X,Y$ be two random variables jointly distributed according to $\Pm_{XY}$. Let $\Pm_X$ be the corresponding marginal of $X$ (\textit{i.e.}, given a measurable set $A$, $\Pm_X(A)=\Pm_{XY}(A\times \mathcal{Y})$) and let $\Q_Y$ be any probability measure over $\mathcal{Y}$. Let $\alpha>0$, the Sibson's Mutual Information of order $\alpha$ between $X,Y$ is defined as:
    \begin{align}
        I_\alpha(X,Y) = \min_{Q_Y} D_\alpha(\Pm_{XY}\|\Pm_X Q_Y). \label{iAlphaDef}
    \end{align}
\end{Definition}
The following, alternative formulation is also useful \cite{verduAlpha}:
\begin{align}
    I_\alpha(X,Y)&= \frac{\alpha}{\alpha-1}\log\mathbb{E}\left[\mathbb{E}^{\frac{1}{\alpha}}\left[\frac{\Pm_{Y|X}}{\Pm_Y}\bigg|Y \right]\right] \\
    &= D_\alpha(\Pm_{XY}\|\Pm_X \Pm_Y) - D_\alpha(\Pm_{Y_\alpha}\|\Pm_Y),
\end{align}
where $\Pm_{Y_\alpha}$ is the measure minimizing \eqref{iAlphaDef}. In analogy with the limiting behavior of $\alpha-$Divergence we have that $\lim_{\alpha\to 1}I_\alpha(X,Y)=I(X;Y)$ while, when $\alpha\to\infty$ we retrieve the following object: $$I_\infty(X,Y)=\log\mathbb{E}_{\Pm_Y}\left[\sup_{x:\Pm_X(x)>0} \frac{\Pm_{XY}(x,Y)}{\Pm_X(x)\Pm_Y(Y)}\right].$$
To conclude, let us list some of the properties of the measure:
\begin{Proposition}[\cite{verduAlpha}]
\label{sibsProperties}
\noindent
\begin{enumerate}
    \item \textbf{Data Processing Inequality}:
    given $\alpha>0$, 
    $I_{\alpha}(X,Z) \leq \min\{I_{\alpha}(X,Y),I_{\alpha}(Y,Z)\}$ if the Markov Chain $X-Y-Z$ holds;
    \item $I_\alpha(X,Y)\geq 0$ with equality iff $X$ and $Y$ are independent;
    \item Let $\alpha_1\leq \alpha_2$ then $I_{\alpha_1}(X,Y)\leq I_{\alpha_2}(X,Y)$;
    \item Let $\alpha\in (0,1)\cup(1,\infty)$, for a given $\Pm_X$, $\frac{1}{\alpha-1}\exp\left(\frac{\alpha-1}{\alpha} I_\alpha(X,Y)\right)$ is convex in $\Pm_{Y|X}$;
    \item $I_\alpha(X,Y) \leq \min\{\log |X |, \log |Y|\}$;
\end{enumerate}
\end{Proposition}
For an extensive treatment of Sibson's $\alpha$-MI we refer the reader to \cite{verduAlpha}.
\subsubsection{Maximal Leakage}\label{maximalLeakage}
A particularly relevant dependence measure, strongly connected to Sibson's Mutual Information is the maximal leakage, denoted by $\ml{X}{Y}.$ It was introduced as a  way of measuring the leakage of information from $X$ to $Y$, hence the following definition:
\begin{Definition}[Def. 1 of \cite{leakage}]\label{leakage}
	Given a joint distribution $\Pm_{XY}$ on finite alphabets $\mathcal{X}$ and $\mathcal{Y}$, the maximal leakage from $X$ to $Y$ is defined as:
	\begin{equation}  \ml{X}{Y}= \sup_{\substack{U-X-Y-\hat{U}}} \log \frac{\mathbb{P}(\{U=\hat{U}\})}{\max_{u\in\mathcal{U}} \mathbb{P}_U(\{u\})},\end{equation}
	where $U$ and $\hat{U}$ take values in the same finite, but arbitrary, alphabet. 
\end{Definition} \noindent
It is shown in~\cite[Theorem 1]{leakage} that, for finite alphabets: 
\begin{align} \label{leakageFormula}
\ml{X}{Y} = \log \sum_{y\in \mathcal{Y}} \max_{\substack{x\in\mathcal{X}}: \Pm_{X}(x)>0} P_{Y|X}(y|x). \end{align}
If $X$ and $Y$ have a jointly continuous pdf $f(x,y)$, we get~\cite[Corollary 4]{leakageLong}:
\begin{align}\ml{X}{Y} = \log \int_{\mathbb{R}} \sup_{x: f_X(x)>0} f_{Y|X}(y|x) dy. \label{leakageContinous} \end{align}
One can show that $\ml{X}{Y}= I_\infty(X;Y)$ i.e., Maximal Leakage corresponds to the Sibson's Mutual Information of order infinity.  This allows the measure to retain the properties listed in Proposition \ref{sibsProperties}, furthermore:
\begin{Lemma}[\cite{leakage}]
	\label{leakageProps}
	For any joint distribution $\Pm_{XY}$ on finite alphabets $\X$ and $\Y$, $\ml{X}{Y} \geq I(X; Y )$.
\end{Lemma}
Another relevant notion, important for its application to Adaptive Data Analysis, is Conditional Maximal Leakage: 
\begin{Definition}[Conditional Maximal Leakage \cite{leakageLong}]
	Given a joint distribution $P_{XYZ}$ on alphabets $\mathcal{X}, \mathcal{Y}, \text{ and } \mathcal{Z}$, define:
	\begin{equation} \Le (X \!\! \to \!\! Y |Z) =\sup_{\substack{U:U-X-Y|Z}} \log\frac{\mathbb{P}(\{U = \hat{U}(Y, Z)\})}{\mathbb{P}(\{U=\tilde{U}(Z)\})},\end{equation}
	where $U$ takes value in an arbitrary finite alphabet and we consider $\hat{U}, \tilde{U}$ to be the optimal estimators of $U$ given $(Y,Z)$ and $Z$, respectively.
\end{Definition}
\noindent Again, it is shown in~\cite{leakageLong} that for discrete random variables $X,Y,Z$:
\begin{equation} \Le (X \!\! \to \!\! Y |Z) = \log \left( \max_{\substack{z:P_Z(z)>0}} \sum_{y} \max_{\substack{x: P_{X|Z}(x|z)}>0} P_{Y|XZ}(y|xz)\right),  \label{conditionalLeakage}
\end{equation}
and \begin{equation}\ml{X}{(Y,Z)} \leq \ml{X}{Y} + \Le (X \!\! \to \!\! Z |Y). \label{ineqLeak}
\end{equation}

\subsubsection{$f-$Mutual Information}\label{fDivergence}
Another generalization of the KL-Divergence can be obtained by considering a generic convex function $f:\mathbb{R}^+\to \mathbb{R}$, usually with the simple constraint that $f(1)=0$. The constraint can be ignored as long as $f(1)<+\infty$ by simply considering a new mapping $g(x) = f(x) - f(1)$. 
\begin{Definition}
	Let $(\Omega,\F,\Pm),(\Omega,\F,\Q)$ be two probability spaces. Let $f:\mathbb{R}^+\to \mathbb{R}$ be a convex function such that $f(1)=0$. Consider a measure $\mu$ such that $\Pm\ll\mu$ and $\Q\ll\mu$. Denoting with $p,q$ the densities of the measures with respect to $\mu$, the $f-$Divergence of $\Pm$ from $\Q$ is defined as follows:
	\begin{align}
	D_f(\Pm\|\Q)=\int q f\left(\frac{p}{q}\right) d\mu.
	\end{align}
\end{Definition}
Despite the fact that the definition uses $\mu$ and the densities with respect to this measure, it is possible to show that $f-$divergences are actually independent from the dominating measure \cite{fDiv1}. Indeed, when absolute continuity between $\Pm,\Q$ holds, i.e. $\Pm\ll\Q$, an assumption we will often use, we retrieve the following \cite{fDiv1}:
\begin{equation}
    D_f(\Pm\|\Q)= \int f\left(\frac{d\Pm}{d\Q}\right)d\Q.
\end{equation}Denoting with $\F_X$ the Sigma-field generated from the random variable $X$, (i.e., $\sigma(X)$), $f$-mutual information is defined as follows:
\begin{Definition}
    Let $X$ and $Y$ be two random variables jointly distributed according to $\Pm_{XY}$ over the a measurable space $(\X\times\Y, \F_{XY})$. 
 Let  $(\X,\F_{X},\Pm_{X}),(\Y,\F_{Y},\Pm_Y)$ be the corresponding probability spaces induced by the marginals.  Let $f:\mathbb{R}^+\to \mathbb{R}$ be a convex function such that $f(1)=0$. The $f-$Mutual Information between $X$ and $Y$ is defined as:
		\begin{equation}
	I_f(X;Y)=D_f(\Pm_{XY}\|\Pm_X\Pm_Y).
	\end{equation}
	If $\Pm_{XY}\ll\Pm_X\Pm_Y$ we have that:
	\begin{equation}
	    I_f(X;Y) = \int f\left(\frac{d\Pm_{XY}}{d\Pm_X\Pm_Y}\right)d\Pm_X\Pm_Y.
	\end{equation}
\end{Definition}
It is possible to see that, if $f$ satisfies $f(1)=0$ and it is strictly convex at $1$, then  $I_f(X;Y)=0$ if and only if $X$ and $Y$ are independent \cite{fDiv1}.
This generalization includes the KL (by simply setting $f(t)=t\log(t)$) and allows to retrieve $\alpha-$Divergences through a one-to-one mapping. But it also includes many more divergences:
\begin{itemize}
	\item Total Variation distance, with $f(t)=\frac12|t-1|$;
	\item Hellinger distance, with $f(t)=(\sqrt{t}-1)^2$;
	\item Pearson $\chi^2$-divergence, with $f(t)=(t-1)^2$.
\end{itemize}
Exploiting a bound involving $I_f(X;Y)$ for a broad enough set of functions $f$ allows to differently measure the dependence between $X$ and $Y$ and it may help us circumventing issues that commonly used measures, like Mutual Information, may suffer from. Consider for instance the following example \cite{genErrWassDist}: let $S$ be a random vector, via Strong Data-Processing inequalities it is possible to show that, given the Markov Chain $S-H-Y$, where $\|H\|\leq k$ and $Y=H+N$ with $N$ Gaussian noise, the Total Variation distance between the joint and the product of the marginals of $S,Y$ is strictly less than $1$, while $I(S;Y)$ may still be infinite. Furthermore, as presented in \cite{convergenceDiv}, different divergences between distributions can provide different convergence rates. It has been proved in \cite{Su1995} that it is possible to construct a random walk that converges in $2n\log n$ steps under KL, $n^2\log n$ steps under the $\chi^2-$distance and $n\log n$ in total variation. This shows that even though several $f-$ divergences may go to $0$ with the number of steps (or samples, in the case of a generalization error bound), the rate of convergence obtainable can be quite different and this can possibly impact the sample complexity in the problems we will analyze in later sections.
\subsection{Orlicz functions and Luxemburg norms}
Let $(\Omega,\mathcal{F}, \mu)$ be a complete and $\sigma$-finite measure space and denote with $L^0(\mu)$ the space of all the $\mathcal{F}$-measurable and real valued functions on $\Omega$. Given an Orlicz function, i.e., a convex function $\psi:[0,+\infty)\to [0,+\infty]$ that vanishes at $0$ and is not identically $0$ or $+\infty$ over the positive real line we can define a functional $I_\psi : L^0(\mu) \to [0,+\infty]$ as $I_\psi(x) = \int_\Omega \psi(|x(t)|)d\mu(t).$ An Orlicz space can then be defined to be \cite{OrliczAmemiyaNorm}:
\begin{equation}
    L_\psi(\mu) = \{x\in L^0(\mu) : I_\psi(\lambda x) < +\infty\text{ for some }\lambda>0\}.
\end{equation}
The Orlicz space is a Banach space (a complete normed vector space) that can be endowed with several norms: the Luxemburg, Orlicz and Amemiya norm. It can also be showed that Amemiya norms are equivalent to Orlicz norms in general \cite{OrliczAmemiyaNorm}. For the purposes of these paper, let us restrict ourselves to probability spaces and define the corresponding norms with respect to random variables and the expectation operator. In particular, let $U$ be an $\mathcal{F}-$measurable random variable, we can define the Luxemburg norm of $U$ with respect to $\mu$:
\begin{equation}
\lVert U\rVert_\psi^\mu = \inf \left\{ \sigma > 0 : \mathbb{E}_\mu\left[\psi\left(\frac{|U|}{\sigma}\right)\right]\leq 1\right\}
\end{equation}
and the Amemiya norm of $U$ with respect to $\mu$:
\begin{equation}
    \lVert U\rVert^{A,\mu}_\psi = \inf \left\{ \frac{\mathbb{E}_\mu\left[\psi(t|U|)\right]+1}{t} : t > 0 \right\}.
\end{equation}
When the measure is not clearly specified it corresponds to the probability measure used to define the space where the random variable lives, although we will often need to be more explicit, as the measures will be often changed.
Given these two quantities and the following definition of convex conjugation, one can show the following generalisation of H\"older's inequality:
\begin{Definition} \label{def:conv-conj}
    Given a convex function $\psi: [0,+\infty) \to \mathbb{R}$, define $\psi^\star: [0,+\infty) \to \mathbb{R}$ as
    \begin{align} \label{eq:def-conv-conj}
        \psi^\star(x) = \sup_{\lambda > 0} \lambda x - \psi(\lambda).
    \end{align}
\end{Definition}
\begin{Lemma}[\cite{jiao2017dependence}]\label{generalisedHolder}
Let $\psi$ be an Orlicz fucntion and $\psi^\star$ denote its conjugate, then for every couple of random variable $U,V$:
    $$\mathbb{E}[UV] \leq \lVert U\rVert_\psi \lVert V\rVert^A_{\psi^\star}.$$
\end{Lemma}
With $\psi(t)=t^\alpha/\alpha$ (and, consequently, $\psi^*(t)=t^\gamma/\gamma$, with $\frac1\gamma+\frac1\alpha = 1$) one recovers H\"older's inequality.
For completeness we included a proof of Lemma \ref{generalisedHolder} in Appendix \ref{sec:GenHoldIneqProof}.
\subsection{Learning Theory}\label{learning}
In this section we will provide some basic background knowledge on learning algorithms and concepts like generalization error. We are mainly interested in supervised learning, where the algorithm learns a \emph{classifier} by looking at points in a proper space and the corresponding labels. \\
More formally, suppose we have an instance space $\mathcal{Z}$ and a hypothesis space $\mathcal{H}$. The hypothesis space is a set of functions that, given a data point $s\in \mathcal{Z}$ outputs the corresponding label $\mathcal{Y}$. Suppose we are given a training data set $\mathcal{Z}^n\ni S =\{z_1,\ldots,z_n\}$ made of $n$ points sampled in an i.i.d. fashion from some distribution $\mathcal{P}$. Given some $n\in\mathbb{N}$, a learning algorithm is a (possibly stochastic) mapping $\mathcal{A}:\mathcal{Z}^n\to \mathcal{H}$ that given as an input a finite sequence of points $S\in\mathcal{Z}^n$ outputs some classifier $h=\mathcal{A}(S)\in\mathcal{H}$.
In the simplest setting we can think of $\mathcal{Z}$ as a product between the space of data points and the space of labels \textit{i.e.}, $\mathcal{Z}=\mathcal{D}\times\mathcal{C}$ and suppose that $\mathcal{A}$ is fed with $n$ pairs data-label $(d,c)\in\mathcal{Z}$. In this work we will view $\mathcal{A}$ as a family of conditional distributions $\mathcal{P}_{H|S}$ and provide a stochastic analysis of its generalization capabilities using the information measures presented so far.
The goal is to generate a hypothesis $h:\mathcal{D}\to \mathcal{C}$ that has good performance on both the training set and newly sampled points from $\mathcal{X}$. In order to ensure such property, the concept of generalization error is introduced.
\begin{Definition} Let $\mathcal{P}$ be some distribution over $\mathcal{Z}$. Let $\ell:\mathcal{H}\times\mathcal{Z}\to\mathbb{R}$ be a loss function. The error (or risk) of a prediction rule $h$ with respect to $\mathcal{P}$ is defined as \begin{equation}L _\mathcal{P}(h)=\mathbb{E}_{Z\sim \mathcal{P}}[\ell(h,Z)],\end{equation}
	while, given a sample $S=(z_1,\ldots,z_n)$, 
	the empirical error
	of $h$ with respect to $S$ is defined as \begin{equation}\label{genEmpRisk}L_{S}(h) = \frac1n \sum_{i=1}^n \ell(h, z_i).\end{equation}
	Moreover, given a learning algorithm $\mathcal{A}:\mathcal{Z}^n\to\mathcal{H}$, its generalization error with respect to $S$ is defined as: \begin{equation}\label{generr}\text{gen-err}_\mathcal{P}(\mathcal{A},S)=|L_{\mathcal{P}}(\mathcal{A}(S))-L_{S}(\mathcal{A}(S))|.\end{equation}
\end{Definition}
The definition just stated considers general loss functions. An important instance for the case of supervised learning is the $0-1$ loss. Suppose again that $\Z=\mathcal{D}\times\mathcal{C}$ and that $\mathcal{H}=\{h|h:\mathcal{D}\to\mathcal{C}\}$, given a couple $(d,c)\in\Z$ and a hypothesis $h:\mathcal{D}\to\mathcal{C}$ the loss is defined as follows: \begin{equation}\ell(h,(d,c))=\mathbbm{1}_{h(d)\neq c},\label{01loss}\end{equation} and the corresponding errors become: \begin{equation}L_\mathcal{P}(h)= \E_{(d,c)\sim\mathcal{P}}[\mathbbm{1}_{h(d)\neq c}] = \mathbb{P}(h(d)\neq c).\end{equation} and \begin{equation}L_S(h)= \frac1n \sum_{i=1}^n \mathbbm{1}_{h(d_i)\neq c_i}.\label{empRisk}\end{equation}


\section{General Results}\label{mainResults}

In this section, we present our main results. First, we prove three general bounds on the probability of an event $E$ under a joint distribution $\Pm_{XY}$ with respect to its probability under the product of the marginals, using notions of Luxumburg norms, Amemiya norms, and $f$-mutual information. We subsequently derive several interesting corollaries that employ common information measures such as Sibson mutual information (section~\ref{sec:SMI}), maximal leakage (section~\ref{sec:ML}),  $\alpha$-divergences and Hellinger divergences (section~\ref{sec:falphadiv}). A particular focus will be given to the bound using maximal leakage, for reasons discussed in the corresponding subsection.\\

Our first main bound employs the Luxemburg and Amemiya norms.
\begin{Theorem}\label{orliczAmemiyaBound}
Let $(\X\times\Y,\F,\Pm_{XY}),(\X\times\Y,\F,\Pm_X\Pm_Y)$ be two probability spaces, and assume that $\Pm_{XY}\ll\Pm_X\Pm_Y$. Given $E\in \F$ and two Orlicz functions $\psi,\varphi$:
\begin{equation}
     \Pm_{XY}(E) \leq \left\lVert \left\lVert \mathbbm{1}_{\{X\in E_Y\}} \right\rVert^{\Pm_X}_\varphi\right\rVert^{\Pm_Y}_\psi \left\lVert\left\lVert \frac{d\Pm_{XY}}{d\Pm_X\Pm_Y}\right\rVert^{A,\Pm_X}_{\varphi^\star}\right\rVert^{A,\Pm_Y}_{\psi^\star},
\end{equation}
where $\mathbbm{1}_{\{\}}$ is the indicator function, and for each $y \in \Y$, $E_y := \{ x: (x,y) \in E \}$ (\textit{i.e.}, the ``fiber'' of $E$ with respect to $y$), and $\varphi^\star$ and $\psi^\star$ are, respectively, the Legendre-Fenchel duals of $\varphi$ and $\psi$.
\end{Theorem}
\begin{proof}
\begin{align}
    \Pm_{XY}(E)= \E_{\Pm_{XY}}[\mathbbm{1}_E]
    &=\E_{\Pm_X\Pm_Y}\left[\mathbbm{1}_E \frac{d\Pm_{XY}}{d\Pm_X\Pm_Y}\right] \\
    &=\E_{\Pm_Y}\left[\E_{\Pm_X} \left[\mathbbm{1}_{\{X\in E_Y\}}\frac{d\Pm_{XY}}{d\Pm_X\Pm_Y}\right]  \right] \\
    & \stackrel{\text{(a)}} {\leq} \E_{\Pm_Y}\left[\left\lVert \mathbbm{1}_{\{X\in E_Y\}} \right\rVert^{\Pm_X}_\varphi\left\lVert \frac{d\Pm_{XY}}{d\Pm_X\Pm_Y}\right\rVert^{A,\Pm_X}_{\varphi^\star} \right] \\
    & \stackrel{\text{(b)}}  {\leq} \left\lVert \left\lVert \mathbbm{1}_{\{X\in E_Y\}} \right\rVert^{\Pm_X}_\varphi\right\rVert^{\Pm_Y}_\psi \left\lVert\left\lVert \frac{d\Pm_{XY}}{d\Pm_X\Pm_Y}\right\rVert^{A,\Pm_X}_{\varphi^\star}\right\rVert^{A,\Pm_Y}_{\psi^\star}, 
\end{align}
where  (a) and (b) follow from Lemma~\ref{generalisedHolder}, \textit{i.e.}, generalised H\"older's inequality.
\end{proof}

Before investigating special cases of the above theorem (yielding explicit bounds in terms of known information measures), we prove a second result that is in the desired form of equation~\eqref{generalBound} and only employs the Luxemburg norm:

\begin{Theorem}\label{luxemburgNormBound}
Let $(\X\times\Y,\F,\Pm_{XY}),(\X\times\Y,\F,\Pm_X\Pm_Y)$ be two probability spaces, and assume that $\Pm_{XY}\ll\Pm_X\Pm_Y$. Given $E\in \F$ and an Orlicz function $\psi$:
\begin{equation}
     \Pm_{XY}(E) \leq \Pm_X\Pm_Y(E)\left({\psi^{\star\star}}\right)^{-1}\left(\frac{1}{\Pm_X\Pm_Y(E)}\right)\left\lVert \frac{d\Pm_{XY}}{d\Pm_X\Pm_Y} \right\rVert^{\Pm_X\Pm_Y}_\psi,
\end{equation}
where for $t \geq 0$, $\psi^{-1}(t):=\inf \{s \geq 0: \psi(s) > t\}$ (and $\psi^{\star \star}$ is defined by applying the transformation in~\eqref{eq:def-conv-conj} twice). 
\end{Theorem}
\begin{proof}
Let $\psi^\star: [0, \infty) \rightarrow [0, \infty)$ be the transform of $\psi$ defined as follows
\begin{align}
\psi^\star(t) =  \sup_{\lambda > 0} \lambda t - \psi(\lambda).
\end{align}
Since $\psi(0)=0$ and $\psi$ is non-negative, it follows that $\psi^\star(0)=0$.
Now, given any $\sigma > 0$ and $t>0$:
\begin{align}
     \Pm_{XY}(E) &= \mathbb{E}_{\Pm_X\Pm_Y}\left[\frac{1}{\sigma} \sigma\mathbbm{1}_E \frac{\frac{d\Pm_{XY}}{d\Pm_{X}\Pm_{Y}}}{t}t\right] \\
     & \stackrel{\text{(a)}} \leq \frac{t}{\sigma}\mathbb{E}_{\Pm_X\Pm_Y}\left[ \psi^\star(\sigma \mathbbm{1}_E) + \psi\left(\frac{\left|\frac{d\Pm_{XY}}{d\Pm_{X}\Pm_{Y}}\right|}{t}\right)\right] \\
     & \stackrel{\text{(b)}} \leq \frac{t}{\sigma }\left(\psi^\star(\sigma) \Pm_{X}\Pm_{Y}(E) +  \mathbb{E}_{\Pm_X\Pm_Y}\left[\psi\left(\frac{\left|\frac{d\Pm_{XY}}{d\Pm_{X}\Pm_{Y}}\right|}{t}\right)\right]\right),
\end{align}
where (a) follows from Young's inequality, and (b) follows from the fact that $\psi^\star(\sigma \mathbbm{1}_E) = \mathbbm{1}_E\psi^\star(\sigma)$ since $\psi^\star(0)=0$.
Now, by choosing $t = \left\lVert \frac{d\Pm_{XY}}{d\Pm_{X}\Pm_{Y}} \right\rVert^{\Pm_X\Pm_Y}_\psi$, we have:
\begin{align}
     \Pm_{XY}(E) 
     &\leq \left\lVert \frac{d\Pm_{XY}}{d\Pm_{X}\Pm_{Y}} \right\rVert^{\Pm_X\Pm_Y}_\psi \frac{\psi^\star(\sigma) \Pm_{X}\Pm_{Y}(E) + 1}{\sigma} \label{generalIneqOrlicz}.
\end{align}
Inequality \eqref{generalIneqOrlicz} holds for every $\sigma>0$, hence we can say that: \begin{align}
    \Pm_{XY}(E) &\leq  \Pm_{X}\Pm_{Y}(E)\cdot \inf_{\sigma >0} \frac{\psi^\star(\sigma) + \frac{1}{\Pm_{X}\Pm_{Y}(E)}}{\sigma}\cdot\left\lVert \frac{d\Pm_{XY}}{d\Pm_{X}\Pm_{Y}} \right\rVert^{\Pm_X\Pm_Y}_\psi  \\
    &= \Pm_{X}\Pm_{Y}(E)\cdot{\left(\psi^{\star\star}\right)}^{-1}\left(\frac{1}{\Pm_{X}\Pm_{Y}(E)}\right)\cdot \left\lVert \frac{d\Pm_{XY}}{d\Pm_{X}\Pm_{Y}} \right\rVert^{\Pm_X\Pm_Y}_\psi \label{lastStepOrlicz},
\end{align}
where \eqref{lastStepOrlicz} follows from~\cite[Lemma 2.4]{BLM2013Concentration} (Note that our $\psi^\star$ plays the role of $\psi$ in Lemma 2.4 and our $\psi^{\star \star}$ plays the role of $\psi^\star$, and the assumption that $\psi^{\star'}(0)=0$ can be replaced by assuming $\psi^\star$ is non-decreasing which holds true in our case).
\end{proof}
\begin{Remark}
The assumption that $\psi$ is convex, non-decreasing, and non-constant implies that $\psi$ is unbounded, so that $\psi^{-1}$ is well-defined for any $t$.
\end{Remark}

\begin{Remark}
With respect to Eq. \eqref{generalBound} we have that $\Pm = \Pm_{XY}, \Q = \Pm_X\Pm_Y$, $f(x) = x{\left(\psi^{\star\star}\right)}^{-1}(1/x)$ and $g(x) = \left\lVert x \right\rVert^{\Q}_\psi$. However, Theorem~\ref{luxemburgNormBound} can be applied to any pair of distributions $\Pm$ and $\mathcal{Q}$ (which do not necessarily correspond to a joint distribution and the product of its marginals). 
\end{Remark}
\begin{Remark}
Note that we defined $\psi^{\star \star}(\lambda)$ as $\sup_{t > 0} \{ \lambda t - \psi^{\star}(t) \}$. If the supremum was over all $t \in \mathbb{R}$, then we would recover $\psi^{\star \star} = \psi$, but equation~\eqref{lastStepOrlicz} would not necessarily hold. Nevertheless, it is often the case for functions of interest that $\psi^{\star \star}$ (as defined) is equal to $\psi$, so that the bound becomes
\begin{equation}
    \Pm_{XY}(E) \leq \Pm_X\Pm_Y(E)\psi^{-1}\left(\frac{1}{\Pm_X\Pm_Y(E)}\right)\left\lVert \frac{d\Pm_{XY}}{d\Pm_X\Pm_Y} \right\rVert^{\Pm_X\Pm_Y}_\psi.
\end{equation}
\end{Remark}

\begin{Theorem} \label{fDivBound}
 Let $\phi : [0, +\infty) \to \mathbb{R}$ be a convex function such that $\phi(1)=0$, and assume $\phi$ is non-decreasing on $[0,+\infty)$. Suppose also that $\phi$ is unbounded, \textit{i.e.}, 
 the generalized inverse, defined as $\phi^{-1}(y)= \inf\{t \geq 0 : \phi(t) > y\}$,  exists.  Given an event $E\in\F$, we have that:
	\begin{align}
	\Pm_{XY}(E) \leq &\Pm_X\Pm_Y(E)\cdot \phi^{-1}\left(\frac{I_\phi(X,Y)+(1-\Pm_X\Pm_Y(E))\phi^\star(0)}{\Pm_X\Pm_Y(E)}\right), \label{generalBoundInv}
	\end{align}
	where $\phi^\star$ is the Legendre-Fenchel dual of $\phi$. Moreover, if $\phi^\star(0) \leq 0$, the bound simplifies to
		\begin{equation}
	\Pm_{XY}(E)\leq \Pm_X\Pm_Y(E)\cdot \phi^{-1}\left(\frac{I_\phi(X,Y)}{\Pm_X\Pm_Y(E)}\right). \label{fDivSimpleBound}
	\end{equation}
\end{Theorem}
\begin{proof}
Let us denote with $p = \Pm_{XY}(E),q =\Pm_X\Pm_Y(E), \bar{p} = 1-\Pm_{XY}(E), \bar{q}=1-\Pm_X\Pm_Y(E) $.
For every $y\geq 0$ we have 
\begin{align}
    I_\phi(X,Y)= D_\phi(\Pm_{XY}\|\Pm_X\Pm_Y) &\overset{\setlabel{dataProcFDiv}}{\geq} D_\phi(\text{Ber}(p)\|\text{Ber}(q))
    \\ &= q\phi\left(\frac{p}{q}\right) + \bar{q}\phi\left(\frac{\bar{p}}{\bar{q}}\right)
    \\ &\overset{\setlabel{youngIneq2}}{\geq} q\phi\left(\frac{p}{q}\right) + \bar{q}\left(\frac{\bar{p}}{\bar{q}}y - \phi^\star(y)\right) \label{DPIafterYoungs}
\end{align}
where \reflabel{dataProcFDiv} follows from the Data-Processing Inequality for $f-$divergences and \reflabel{youngIneq2} follows from Young's inequality.
Choosing $y=0$ in \eqref{DPIafterYoungs} and re-arranging the terms we retrieve
\begin{align}
    &\frac{I_\phi(X,Y) + \bar{q}\phi^\star(0)}{q} \geq \phi\left(\frac{p}{q}\right) \iff \label{stepBeforeLastFDiv} \\
    &q\phi^{-1}\left(\frac{I_\phi(X,Y) + \bar{q}\phi^\star(0)}{q}\right) \geq p. \label{LastStepFDiv}
\end{align}
\end{proof}

\begin{Remark} Alternative proofs can be constructed using the variational representation of $\phi$-divergences for a convex function $\phi$ or an approach similar to the proof of Theorem \ref{luxemburgNormBound}. They can be found in Appendix \ref{sec:altProofTheoremFDiv}.
\end{Remark}

While these results are quite general, computing the Luxemburg or the Amemiya norm can be complicated for most functions. Moreover, our purpose is to retrieve, on the right-hand side of $\Pm_{XY}(E)$ some function of $\Pm_X\Pm_Y(E)$ and an information measure (as function of the Radon-Nikodym derivative). With this drive we will now compute some specific instances of this result for certain choices of $\varphi$ and $\psi$, or $\phi$, that allow us to retrieve well-known objects in information theory. While being a specific instance of Theorem~\ref{orliczAmemiyaBound}, the first result we derive will still be quite general and depend on four parameters. Different choices of these parameters give rise to bounds involving different R\'enyi information measures.
\begin{Theorem}\label{alphaExpBound}
Let $(\X\times\Y,\F,\Pm_{XY}),(\X\times\Y,\F,\Pm_X\Pm_Y)$ be two probability spaces, and assume that $\Pm_{XY}\ll\Pm_X\Pm_Y$. Given $E\in \F$ and $y\in\Y$, let $E_y = \{x : (x,y)\in E\}$, \textit{i.e.} the \enquote{fibers} of $E$ with respect to $y$. Then,
\begin{align}
    \Pm_{XY}(E) 
    \leq &\mathbb{E}^{1/\gamma'}_{\Pm_Y}\left[\Pm_X(E_Y)^{\gamma'/\gamma}\right]\mathbb{E}^{1/\alpha'}_{\Pm_Y}\left[\E_{\Pm_X}^{\alpha'/\alpha}\left[ \left(\frac{d\Pm_{XY}}{d\Pm_X\Pm_Y}\right)^\alpha \right]\right], \label{genBoundAlpha}
\end{align}
where $\gamma,\alpha,\gamma',\alpha'$ are such that $1=\frac1\alpha+\frac1\gamma = \frac{1}{\alpha'}+\frac{1}{\gamma'}$.
\end{Theorem} 
\begin{Remark}
    A proof of this result follows from Theorem \ref{orliczAmemiyaBound} choosing $\varphi(t)=\frac{t^\gamma}{\gamma}$ and $\psi(t)=\frac{t^{\gamma'}}{\gamma'}$ with $\gamma,\gamma' \geq 1$. A more explicit proof can be written using the classical H\"older's inequality twice (similarly to the proof of Theorem \ref{orliczAmemiyaBound}): once for $\Pm_X$ and once for $\Pm_Y$.
\end{Remark}
\begin{Remark}
    It is clear from the proof that one can similarly bound $\mathbb{E}[g(X,Y)]$ (instead of $\E[\mathbbm{1}_E]$) for any positive function $g(X,Y)$ that is $\Pm_X\Pm_Y$-integrable. But the shape of the bound becomes more complex as one in general does not have that $g(X,Y)^\gamma=g(X,Y)$ for every $\gamma\geq 1$. 
\end{Remark}

\section{Sibson's Mutual Information} \label{sec:SMI}

Starting from Theorem~\ref{alphaExpBound} and considering the limit as $\alpha'\to1$, which implies $\gamma'\to+\infty$, we retrieve a bound in terms of Sibson mutual information:
\begin{Corollary}\label{sibsMIBoundCor}
Given $E\in\F$, we have that:
\begin{align}
\Pm_{XY}(E)&\leq \left(\esssup_{\Pm_y} \Pm_X(E_Y)\right)^{1/\gamma} \mathbb{E}_{\Pm_Y}\left[\E^{1/\alpha}_{\Pm_X}\left[ \left(\frac{d\Pm_{XY}}{d\Pm_Yd\Pm_X}\right)^\alpha\right]\right] \label{sibsNonVerdu}\\ &=\left(\esssup_{\Pm_y} \Pm_X(E_Y)\right)^{1/\gamma} \exp\left(\frac{\alpha-1}{\alpha} I_{\alpha}(X,Y)\right), \label{sibMIBound}
\end{align}
where $I_\alpha(X,Y)$ is the Sibson mutual information of order $\alpha$
\cite{verduAlpha}, and $\alpha$ and $\gamma$ satisfy $\frac{1}{\alpha} + \frac{1}{\gamma} =1$.
\end{Corollary}
\begin{Remark}
    An in-depth study of $\alpha-$Mutual Information appears in \cite{verduAlpha}, where a slightly different notation is used. For reference, we can restate Eq. \eqref{sibsNonVerdu} in the notation of \cite{verduAlpha} to obtain:
    \begin{align}
        \Pm_{XY}(E)\leq &\left(\esssup_{\Pm_y} \Pm_X(E_Y)\right)^{1/\gamma} \mathbb{E}_{\Pm_Y}\left[\E^{1/\alpha}_{\Pm_X}\left[ \left(\frac{dP_{Y|X}}{d\Pm_Y}\right)^\alpha \bigg| Y \right]\right].
    \end{align}
\end{Remark}
 Given that $\alpha$ and $\gamma$ are H\"older's conjugates, the bound in~\eqref{sibMIBound} can be rewritten as:
\begin{equation} \label{eq:sibsboundalpha}
    \Pm_{XY}(E)\leq \exp\left(\frac{\alpha-1}{\alpha}\left( I_{\alpha}(X,Y)+\log\left(\esssup_{\Pm_y} \Pm_X(E_Y)\right)\right)\right).
\end{equation}
An interesting property of Sibson $\alpha$-Mutual Information is that it is non-decreasing with respect to $\alpha$ \cite{verduAlpha}. Considering the right hand side of~\eqref{eq:sibsboundalpha} we have that, for $\alpha_1\leq \alpha_2$:  
\begin{equation} 
\frac{\alpha_1-1}{\alpha_1}I_{\alpha_1}(X,Y)\leq \frac{\alpha_2-1}{\alpha_2}I_{{\alpha}_2}(X,Y),
\end{equation}  
thus, choosing a smaller $\alpha$ yields a better dependence on $I_\alpha(X,Y)$ in the bound; but given that $\frac{\alpha_1-1}{\alpha_1} \leq \frac{\alpha_2}{\alpha_2-1}$ and $\log \esssup_{\Pm_y} \Pm_X(E_Y) \leq 0$, the second term increases for smaller values of $\alpha$.
This leads to a trade-off between the two quantities. 
We will now explore an interesting application of Corollary \ref{sibsMIBoundCor} that comes from the field of learning theory: generalization error bounds. In such applications, $\Pm_X(E_y)$ is typically exponentially decaying with the number of samples for every $y$. Moreover, a different perspective on generalization error bounds, i.e., sample complexity bounds, allow us to see the trade-off between different values of $\alpha$ more explicitly. 
\subsection{Generalization Error Bounds}
Consider now the learning setup as defined in Section
\ref{learning}.
The next result can be used to give a concentration bound on the generalization error defined in Equation \eqref{generr}:
\begin{Corollary}\label{generrSibs2}
	Let $\mathcal{A}:\Z^n \to \mathcal{H}$ be a learning algorithm that, given a sequence $S$ of $n$ points, returns a hypothesis $h\in \mathcal{H}$. Suppose $S$ is sampled i.i.d according to some distribution $\mathcal{P}$ over $\Z$. Let $\ell:\mathcal{H}\times\Z \to \mathbb{R}$ be a loss function such that $\ell(h,Z)$ is $\sigma^2$-sub-Gaussian random variable for every $h\in\mathcal{H}$.  
	Given $\eta \in (0,1)$, let $E=\{(S,h):|L_{\mathcal{P}}(h)-L_S(h)|>\eta \}$. Fix $\alpha \geq 1$.  Then,
	\begin{align}
	\mathbb{P}(E) \leq 2 \exp\left(\frac{\alpha-1}{\alpha}\left(I_\alpha(S,\A(S)) -n\frac{\eta^2}{2\sigma^2}\right)\right).
	\end{align}
Consequently, 	in order to ensure a confidence of $\delta\in(0,1)$, i.e. $\mathbb{P}(E)\leq \delta$, it is sufficient to have $m$ samples where 
\begin{equation}\label{sampleComplexitySibsMI}
   m\geq \frac{2\sigma^2}{\eta^2}\left(I_\alpha(S,\A(S))+\log2+\frac{\alpha}{\alpha-1}\log\left(\frac{1}{\delta}\right)\right). 
\end{equation}
\end{Corollary}
\begin{Remark}
The corollary applies to the special case in which $\Z = \mathcal{D} \times \mathcal{C}$ and $\ell$ is the $0$-$1$ loss function as defined in~\eqref{01loss}. Indeed, one can show that $\ell$ is $\sigma^2$-sub-Gaussian for $\sigma = \frac{1}{2}$. Moreover, in this case we only need to assume that the samples $S$ are independent, as the use of Hoeffding's inequality in the proof below can be replaced by McDiarmid's inequality (for functions with bounded differences).	
\end{Remark}
Smaller $\alpha$ means that $I_\alpha(S,\A(S))$ will be smaller, but it will imply a worse dependency on $\log(1/\delta)$ in the sample complexity. It is worth noting that, for fixed $\alpha$, the sample complexity dependency on $\log(1/\delta)$ is optimal (up to constants). In particular, consider the setup of PAC learning with finite $\mathcal{H}$ with VC dimension $d$, e.g, assume that $\mathcal{D}=[d]$ and $\mathcal{H}=\{0,1\}^{\mathcal{D}}$, we have that the VC-dimension of $\mathcal{H}$ is $d$ \cite{learningMI,learningBook}. 
By~\cite[Theorem 6.8]{learningBook}, we know that the number of necessary samples for learning, in the realizable case, satisfies $ m \geq c \frac{d+\log(1/\delta)}{\eta}$, for some constant $c$.
Assume also that $\mathcal{A}$ is the ERM algorithm, in which case the generalization error and the true error are the same and   $I_\alpha(S,\mathcal{A}(S))\leq  \log(|\mathcal{H}|)=d$. 
From \eqref{sampleComplexitySibsMI} for the $0-1$ loss we have that $m\geq c\frac{d+\gamma\log(1/\delta)}{\eta^2}$ hence, for a given $\alpha$, the dependency on $\delta$ is optimal. A similar reasoning could be applied to the agnostic case in order to tackle the optimality with respect to $\eta$ as well. 

\begin{proof}[Proof of Corollary~\ref{generrSibs2}]
    Fix $\eta\in(0,1)$. Let us denote with $E_h$ the fiber of $E$ over $h$ for some $h\in\mathcal{H}$, i.e. $E_h=\{S : |L_\Pm(h)-L_S(h)|>\eta\}$.  By assumption we have that $\ell(h,Z)$ is $\sigma^2$-sub-Gaussian for every $h$. We can thus use Hoeffding's inequality 
    and retrieve that for every $h\in\mathcal{H}:$ \begin{equation}
    \mathcal{P}_S(E_h) \leq 2\cdot \exp\left(-n\frac{\eta^2}{2\sigma^2}\right). \label{hoeffdings} 
\end{equation}
    Then it follows from Corollary~\ref{sibsMIBoundCor} and inequality~\eqref{hoeffdings} that:
    \begin{align}
	\mathbb{P}(E) &\leq \exp\left(\frac{\alpha-1}{\alpha}I_\alpha(S, \mathcal{A}(S))\right)\cdot \left(2\exp\left(-n\frac{\eta^2}{2\sigma^2}\right)\right)^\frac{1}{\gamma}
	\\ &= 2 \exp\left(\frac{\alpha-1}{\alpha}\left(I_\alpha(S,\A(S))-n\frac{\eta^2}{2\sigma^2}\right)\right).
	\end{align}
\end{proof}

\section{Maximal Leakage} \label{sec:ML}

An interesting special case of Corollary~\ref{sibsMIBoundCor} is to let $\alpha\to\infty$. In this scenario, in the right-hand side of Eq.~\eqref{sibMIBound} we obtain Maximal Leakage~\cite{leakage}. Maximal Leakage has gained growing interest in the last few years and enjoys a series of properties that are of particular interest to us and we will soon analyze. The result will be thus stated independently.
Note that considering the other extreme, {\it i.e.,} $\alpha\to 1$ we retrieve a trivial bound. Indeed, letting $\alpha\to 1$ in any of our results leads to a bound of 1 on $\Pm_{XY}(E)$. This means that our approach does not provide bounds that exploit either the Kullback-Leibler divergence or the Mutual Information. Nonetheless, we will provide some comparison with an analogous result obtained for Mutual Information (although, through a different approach \cite{learningMI,infoThGenAn}) in Section~\ref{sec:leakageAndMutI}.
\begin{Corollary}\label{adaptML}
Given $E\in\F$, we have that:
\begin{equation}
\Pm_{XY}(E) \leq \left(\esssup_{\Pm_y} \Pm_X(E_Y)\right) \exp\left(\ml{X}{Y}\right). \label{MLBound}
\end{equation}
\end{Corollary}
\begin{proof}
The proof follows directly from Corollary~\ref{sibsMIBoundCor} and by noting that when $\alpha \to \infty$ then $\gamma \to 1$ and $\ml{X}{Y}=I_\infty(X,Y)$ \cite{leakage}.
\end{proof}
\begin{Remark}
    Corollary \ref{adaptML} can also be proven independently going through the equivalent formulation of $D_\infty$~\cite[Theorem 6]{RenyiKLDiv} and the fact that $\exp(\ml{X}{Y})=\mathbb{E}_{\Pm_Y}[\exp(D_\infty(\Pm_{Y|X=x}\|\Pm_Y))]$, c.f.~Appendix \ref{sec:altProofMaximalLeakage}.
\end{Remark}

This result is particularly useful for the following reasons:
\begin{itemize}
    \item Maximal Leakage is more amenable to analysis due to its semi-closed form (e.g., it is possible to easily compute the maximal leakage of noise-addition mechanisms);
    \item The absence of the power $\frac1\gamma$ in \eqref{MLBound} as compared to the right-hand side of~\eqref{sibMIBound} allows us to provide a generalization of the classical concentration of measure results in adaptive scenarios;
    \item A conditional version of Maximal Leakage allows us to provide adaptive composition results (discussed in Section~\ref{adaptiveDA}).
\end{itemize}

Before discussing examples in which we analyze (simple) schemes with maximal leakage, we first discuss the tightness of the bound.

\subsection{Tightness 
} \label{sec:MLtight}

We illustrate the bound by first giving three examples where inequality~\eqref{MLBound} is met with equality for varying scenarios of dependence: $X$ is independent from $Y$, $X$ and $Y$ are equal, and $X$ and $Y$ are related but not equal. 
\begin{Example}[independent case]
\label{tightness1}
   Suppose that $E$ is such that $\Pm_X(E_y)=\zeta$ for all $y\in\Y$. In that case we have that, if $X$ and $Y$ are independent:
    \begin{equation}
        \zeta=\mathbb{E}_{\Pm_Y}[\Pm_X(E_y)]=\Pm_{XY}(E) \leq \zeta.
    \end{equation}
\end{Example}
\begin{Example}[strongly dependent case]
\label{tightness2}
   Consider the example presented in \cite{learningMI}: suppose $X=Y\sim\mathcal{U}([n])$ then we have that $\ml{X}{Y}=\log n$ and if $E=\{(x,y)\in [n]\times[n] | x=y\}$ then, 
    \begin{equation}
        1= \Pm_{XY}(E) \leq \frac1n \cdot n = 1.
    \end{equation}
\end{Example}
\begin{Example}
Suppose $(X,Y)$ is a doubly-symmetric binary source with parameter $p$ for some $p<1/2$. Let  $E = \{(x,y): x=y\}$. Then,
\begin{equation}
1 - p = \Pm_{XY}(E) \leq \frac{1}{2}(2(1-p)) = 1 -p.
\end{equation}
\end{Example}

The above examples show that when the worst-case behavior (i.e., $\max_y\Pm_X(E_y)$) matches with the average-case behavior (i.e., $\mathbb{E}_{\Pm_Y}[\Pm_X(E_y)]=\Pm_X\Pm_Y(E)$), our bound represents a generalization of the classical concentration of measure inequalities for adaptive settings. This is typically the case in learning scenarios of interest, where we generalise Hoeffding's and McDiarmid's inequalities.

Moreover, the following proposition shows that the bound is tight in the following strong sense: if we want to bound the ratio $\Pm_{XY}(E)/(\esssup_{\Pm_Y} \Pm_X(E_y))$ as a function of $\Pm_{Y|X}$ only (i.e., independently of $\Pm_X$ and $E$), then $\exp \{ \ml{X}{Y} \}$ is the best bound we could get:

\begin{Proposition} \label{prop:MLboundtightness}
Given finite alphabets $\X$ and $\Y$, and a  fixed conditional distribution $\Pm_{Y|X}$, then there exists $\Pm_X$ and $E$ such that~\eqref{MLBound} is met with equality. That is,
\begin{align}
\label{eq:MLboundtightness}
\sup_{ E \subseteq \X \times \Y} \sup_{\Pm_X} \log \frac{\Pm_{XY}(E)}{\esssup_{\Pm_Y} \Pm_X(E_Y) } = \ml{X}{Y}.
\end{align}
\end{Proposition}

\begin{proof}
 Define a function $f: \Y \rightarrow \X$ such that $f(y) \in \argmax_{x \in \X} \Pm_{Y|X} (y|x)$, and let $\X_f \subseteq \X$ be the image of $f$. Now, let $\Pm_X$ be the uniform distribution over $\X_f $, and $E = \{ (x,y): x=f(y) \}$. Then, for any $y \in \Y$, 
\begin{align} \label{eq:pxey}
 E_y = \{ f(y) \} \Rightarrow \Pm_X(E_y)  = \frac{1}{|\X_f|}. 
\end{align}
So we get 
\begin{align}
\Pm_{XY} (E) & = \sum_{(x,y) \in E} \Pm_{XY}(x,y) \\
& = \sum_{y \in \Y} \sum_{x \in E_y} \Pm_{X} (x) \Pm_{Y|X} (y|x) \\
& = \sum_{y \in \Y} \Pm_X(f(y)) \Pm_{Y|X} (y| f(y) ) \\
& = \frac{1}{|\X_f|} \sum_{y \in \Y} \max_x \Pm_{Y|X} (y|x),
\end{align}
where the last equality follows from~\eqref{eq:pxey} and the definition of $f$.
\end{proof}

\subsection{Generalization Error Bounds} \label{sec:MLgenerror}

We will now explore how this result can be applied in providing bounds on the generalization error of learning algorithms.

\begin{Corollary} \label{generrML} 
	Let $\mathcal{A}:\Z^n \to \mathcal{H}$ be a learning algorithm that, given a sequence $S$ of $n$ points, returns a hypothesis $h\in \mathcal{H}$. Suppose $S$ is sampled i.i.d according to some distribution $\mathcal{P}$ over $\Z$. Let $\ell:\mathcal{H}\times\Z \to \mathbb{R}$ be a loss function such that $\ell(h,Z)$ is $\sigma^2$-sub-Gaussian random variable for every $h\in\mathcal{H}$.  
	Given $\eta \in (0,1)$, let $E=\{(S,h):|L_{\mathcal{P}}(h)-L_S(h)|>\eta \}$. Then,
	\begin{align}
	\mathbb{P}(E) \leq 2\cdot\exp\left(\ml{S}{\mathcal{A}(S)} -n\frac{\eta^2}{2\sigma^2}\right).
	\end{align}
Consequently, in order to ensure a confidence of $\delta\in(0,1)$, i.e. $\mathbb{P}(E)\leq \delta$, it is sufficient to have $m$ samples where 
\begin{equation}\label{sampleComplexityML}
   m\geq \frac{2\sigma^2}{\eta^2}\left(\ml{S}{\mathcal{A}(S)}+\log\left(\frac{2}{\delta}\right)\right). 
\end{equation}
\end{Corollary}
The proof follows from Corollary~\ref{adaptML} and the same technique used to prove Corollary~\ref{generrSibs2}.
\begin{Remark}
Similarly to Corollary~\ref{generrSibs2}, this bound applies to the case in which $\ell$ is the 0-1 loss function with $\sigma = \frac{1}{2}$. Moreover, as discussed following Corollary~\ref{generrSibs2}, the dependence on $\log(1/\delta)$ is optimal. Indeed, let us consider the very same example as in the previous section. Let $\mathcal{D}=[d]$ and $\mathcal{H}=\{0,1\}^{\mathcal{D}}$, we have that the VC-dimension of $\mathcal{H}$ is $d$. Choosing again $\A$ to be the ERM algorithm we have that $\ml{S}{\A(S)}=d$. Looking at \eqref{sampleComplexityML} with $\sigma = 1/2$, we can see how the lack of the $\alpha/(\alpha-1)$ term (that one can find in \eqref{sampleComplexitySibsMI} instead) allows us to make a clear analogy with the VC-dimension bound stated in \cite[Theorem 6.8]{learningBook}.  More precisely, from \eqref{sampleComplexityML} we have that $m\geq \frac{d+\log(2/\delta)}{2\eta^2}$ while \cite[Theorem 6.8.3]{learningBook} (realizable case) tells us that $m\geq c \frac{d+\log(1/\delta)}{\eta}$ for some constant $c$. 
\end{Remark}

Whenever $\mathcal{A}$ is independent from the samples $S$, we have that $\exp(\ml{S}{\mathcal{A}(S)})=1$ and we immediately fall back to the non-adaptive scenario:
$\mathbb{P}(E)\leq 2\cdot \exp\left(-n\frac{\eta^2}{2 \sigma^2}\right)$ i.e., Hoeffding's inequality.


\subsection{Analyzing Schemes via Maximal Leakage} \label{sec:MLschemes}

A simple way of keeping the Maximal Leakage of an algorithm $\A(X)$ bounded (and thus ensure generalization) is to add noise (e.g., $\hat{Y}=\A(X)+N$ with $\mathcal{A}$ a real-valued function). The proofs for this section can be found in Appendix \ref{app:ANC}.
\begin{Lemma}[Laplacian Noise]\label{laplacianMech}
Let $g: \mathcal{X}^n \to \mathbb{R}$ be a function such that $g(x)\in [a,c], a<c ~\forall x\in\mathcal{X}^n$. The mechanism
$\mathcal{M}(x) = g(x) + N$ where $N \sim Lap(b)$ is such that:
\begin{equation}
    \ml{X}{\mathcal{M}(X)} = \log\left(1+\frac{(c-a)}{b} \right).
\end{equation}
\end{Lemma}
Similar results can be obtained analyzing different types of noise.
\begin{Lemma}[Gaussian Noise]
Let $g : \mathcal{X}^n \to \mathbb{R}$ be a function such that $\forall x\in\mathcal{X}^n ~g(x)\in [a,c], a<c $. The mechanism
$\mathcal{M}(x) = g(x) + N$ where $N \sim \mathcal{N}(0,\sigma^2)$ is such that:
\begin{equation}
    \ml{X}{\mathcal{M}(X)}  = \log\left(1 + \frac{(c-a)}{\sqrt{2\pi \sigma^2}} \right).
\end{equation}
\end{Lemma}
\begin{Lemma}[Exponential Noise]
Let $f : \mathcal{X}^n \to \mathbb{R}$ be a function such that $\forall x\in\mathcal{X}^n ~f(x)\in [a,c], c>0 $. The mechanism
$\mathcal{M}(x) = f(x) + N$ where $N \sim Exp(\lambda)$ (i.e. $\mathbb{E}[N]= (1/\lambda) = b$) is such that:
\begin{equation}
    \ml{X}{\mathcal{M}(X)}  = \log \left( 1+\frac{(c-a)}{b} \right).
\end{equation}
\end{Lemma}
The addition of carefully calibrated noise to control maximal leakage can be used in practice to obtain generalization guarantees of learning algorithms. As an exact analogy to \cite[Corollary 4]{infoThGenAn} we can state the following corollary, involving a noisy version of the Empirical Risk Minimization (ERM) algorithm.
\begin{Corollary}\label{noisyERM}
Let us consider the following algorithm:
    \begin{equation}
        \A(S)= \arg \min_{h\in\mathcal{H}} (L_S(h) + N_h),
    \end{equation}
    where $N_h$ is exponential noise drawn independently from the input, added to the empirical risk of each hypothesis on a given data-set $S$. Suppose $\mathcal{H}$ is countable (i.e., finite or countably infinite), and denote with $N_i$ the noise added to the hypothesis $h_i$ with mean $b_i$.
    Then, for every $\eta\in(0,1)$:
    \begin{equation}
        \mathbb{P}(\text{gen-err}(\A)\geq \eta) \leq 2\exp\left( \sum_{i=1}^{|\mathcal{H}|} \log\left( 1+ \frac{1}{b_i} \right) - 2n\eta^2\right).
    \end{equation}
    Choosing $b_i= i^{1.1}/n^{1/3}$, we retrieve:
    \begin{equation}
        \mathbb{P}(\text{gen-err}(\A)\geq \eta) \leq 2\exp\left( -n(2\eta^2-11/n^{2/3})\right).
    \end{equation}
\end{Corollary}
     This example shows how simply the maximal leakage bound can be used, in contrast with the mutual information one. Indeed, following the proof of \cite[Corollary 4]{infoThGenAn}, the mutual information of the same mechanism analyzed here is hard to compute directly and the quantity $I(S;H)$ is, in the end, effectively upper-bounded using maximal leakage:
     \begin{align}
         I(S;H) &\leq \sum_{i=1}^{|\mathcal{H}|} \log\left(1+\frac{L_\mu(h_i)}{b_i}\right)\\ 
         &\leq\sum_{i=1}^{|\mathcal{H}|} \log\left(1+\frac{1}{b_i}\right) = \ml{S}{H}.
     \end{align}
 \begin{Remark}
     The noisy version of the ERM algorithm does not provide the same guarantees (with respect to the classical ERM) in terms of training error. Having a small generalization error means that the training and testing error are close, but they could both be large. In this case, adding noise we reduce the information measure (by DPI) and as a consequence of our bounds, the generalization error is also reduced. The addition of noise, however, can increase the training error of the new algorithm. In particular, we will no longer choose the empirical loss minimizer ($h^\star = \arg \min_{h\in\mathcal{H}} L_S(h)$) but some hypothesis $h$ that minimizes $L_S(h)+N_h$ and whose error is more or less close to $L_S(h^\star)$, depending on the noise. Hence, while the generalization error may be smaller, both training and testing error could actually be getting larger.
 \end{Remark}

\section{Hellinger and $\alpha$-Divergences} \label{sec:falphadiv}

In this section, we demonstrate new bounds in terms of $\alpha$-Divergences and $\phi$-Divergences, and in particular, Hellinger divergences.

Choosing $\alpha'=\alpha$ and thus $\gamma'=\gamma$ in Theorem~\ref{alphaExpBound}, we retrieve:
\begin{Corollary}\label{alphaDivBound}
Given $E\in\F$ and $\alpha > 1$, we have that:
\begin{equation}
\Pm_{XY}(E)\leq (\Pm_X\Pm_Y(E))^{\frac{\alpha-1}{\alpha}}\exp\left(\frac{\alpha-1}{\alpha}D_\alpha(\Pm_{XY}\|\Pm_X\Pm_Y)\right).
\end{equation}
\end{Corollary}
%
This result can, in fact, be proven in several alternative ways:  
using the data processing inequality for $D_\alpha$, using the bound in Theorem~\ref{luxemburgNormBound} with $\psi(t)= \frac{t^\alpha}{\alpha}$, and using Theorem~\ref{fDivBound} with $\phi_\alpha = (t^{\alpha}-1)/(\alpha-1)$ which gives a bound in terms of Hellinger divergences that are in 1-to-1 mapping with $\alpha$-divergences. These alternative proofs can be found in Appendix~\ref{sec:altProofsCorollaryAlphaDiv}.
%
With respect to Eq. \eqref{generalBound} we again have that $\Pm = \Pm_{XY}$ and $\Q= \Pm_X\Pm_Y$, $f(x) = x^{1/\gamma}$ but in this case $g(\Pm/\Q)$ does involve a divergence. We thus have that $g(x)=\exp\left(\frac{1}{\alpha} \log\mathbb{E}_{\Pm_X\Pm_Y}\left[x^\alpha\right]\right).$
As Hellinger divergences are of independent interest, and include important objects, like the $\chi^2$-divergence, we restate the bound explicitly in terms of Hellinger divergences. Recall, Hellinger divergences can be characterized by $\phi_p(t)= (t^p-1)/(p-1)$ with $p\in (0,1)\cup(1,+\infty)$.  Theorem~\ref{fDivBound} can, though, only be applied to $p\in (1,+\infty)$, as $\phi_p(\cdot)$ is concave for $p \in (0,1)$. Let $\mathcal{H}_p(X,Y)$ denote the Hellinger divergence of $\Pm_{XY}$ from $\Pm_X\Pm_Y$ and with a slight abuse of notation let $\chi^2(X,Y)$ denote $\chi^2(\Pm_{XY} || \Pm_X \Pm_Y)$. We can now state the following.

\begin{Corollary}\label{generalHellingerBound} 
Let $E\subseteq \X\times \Y$ and let $p\in (1,+\infty)$ then $I_{\phi_p}(X,Y) = \mathcal{H}_p(X,Y)$ and
\begin{equation}
    \Pm_{XY}(E) \leq \Pm_X\Pm_Y(E)^{\frac{p-1}{p}}\cdot  \left((p-1) \mathcal{H}_p(X,Y) + 1\right) ^{1/p}.
\end{equation}
In particular, for $p=2$, we have
 \begin{align}
        \Pm_{XY}(E)& \leq \sqrt{(\chi^2(X,Y)+1)\Pm_X\Pm_Y(E)}\\
        & {\leq} \sqrt{\exp{\left(\ml{X}{Y}\right)}\Pm_X\Pm_Y(E)}. \label{chiLeakageBound}
    \end{align}
\end{Corollary}
The last inequality follows from the fact that $\chi^2(X,Y)\leq \exp{\left(\ml{X}{Y}\right)}-1$ (cf.~\cite{ISIT2018}).
Applying this result to a learning environment we get:
\begin{Corollary}\label{hellingerDivBound}
Let $\mathcal{A}:\Z^n \to \mathcal{H}$ be a learning algorithm that, given a sequence $S$ of $n$ points, returns a hypothesis $h\in \mathcal{H}$. Suppose $S$ is sampled i.i.d according to some distribution $\mathcal{P}$ over $\Z$. Let $\ell:\mathcal{H}\times\Z 
	\to \mathbb{R}$ be a loss function such that $\ell(h,Z)$ is a $\sigma^2$-sub-Gaussian random variable, for some $\sigma$ and for every $h\in\mathcal{H}$.  
	Given $\eta \in (0,1)$, let $E=\{(S,h):|L_{\mathcal{P}}(h)-L_S(h)|>\eta \}$ and $p\in (1,+\infty)$.
	Then,
	\begin{align}\mathbb{P}(E) \leq 2^{\frac{p-1}{p}}\exp\left(\frac1p\left(\log((p-1)H_p(S,\A(S))+1)-n\frac{\eta^2}{2\sigma^2}\right)\right).\label{generrPHellinger}
	\end{align}
In particular, 
	\begin{align}
	\mathbb{P}(E) \leq \sqrt{2}\exp\left(\frac12\left(\log(\chi^2(S,\A(S))+1)-n\frac{\eta^2}{2\sigma^2}\right)\right), \label{generrChiSq}
	\end{align}
	and in order to ensure a confidence of $\delta\in(0,1)$, \textit{i.e.} $\mathbb{P}(E)\leq \delta$, it is sufficient to have $m$ samples where  
	\begin{align}
	        m\geq \frac{2 \sigma^2}{ \eta^2} \log(\chi^2(S,\A(S))+1)+2\log\left(\frac{\sqrt{2}}{\delta}\right).
	\end{align}
\end{Corollary}
\begin{Remark}
As before, this result applies to 0-1 loss functions with $\sigma = \frac{1}{2}$.
\end{Remark}
An implication of~\eqref{generrChiSq}, say for 0-1 loss functions, is the following: if $\chi^2(S,\A(S)) < \exp(2n\eta^2) -1$ then we can guarantee an exponential decay in the probability of having a large generalization error.
Doing the same with inequality~\eqref{chiLeakageBound}, one gets:
\begin{equation}
    \mathbb{P}(E)\leq \sqrt{2}\exp\left(\frac12\left(\ml{S}{\A(S)}-2n\eta^2\right)\right).
\end{equation}
Given the relationship between these two measures, one has that every time $\exp(\ml{X}{Y})\leq 2n\eta^2$ then $\chi^2(X,Y) \leq \exp(2n\eta^2)-1$ and thus, generalization with maximal leakage implies generalization with $\chi^2$. An advantage of using $\chi^2(X,Y)$ is that it can be significantly smaller than $\ml{X}{Y}$. Indeed:
\begin{Example}
Let $X\sim\text{Ber}(1/2)$ and let $Y=\text{BSC}(p)$, with $p<1/2$. Thus, $P_{Y|X=x}(x)=1-p$. In this case $\chi^2(X,Y)=(1-2p)^2$ while $\exp(\ml{X}{Y})-1=(1-2p)$. It is easy to see that, since $(1-2p)<1$ then $(1-2p)^2$ can be much smaller than $(1-2p)$.
\end{Example}
On the other hand, an advantage in using maximal leakage is that it depends on $\Pm_X$ only through the support. This allows us to provide bounds that depend only loosely on the distribution over the training samples. $\chi^2(X,Y)$, instead, cannot be computed unless one has \textbf{full} access to $\Pm_X$. Such distributions can be very complicated and typically defined on large dimensional spaces (\textit{e.g.}, images, audio-recordings, etc.). In general, one only has access to (and control over) the conditional distributions $P_{Y|X}$ induced by the chosen learning algorithm. This can render the usage of bounds like Ineq. \eqref{generrChiSq} difficult in practice, although tighter in theory.
Another important characteristic of maximal leakage is that, as a consequence of the chain rule it satisfies, it composes adaptively (more details in Section \ref{adaptiveDA}). Such property is not known to hold, in general, for either $f-$ or Sibson's $\alpha-$Mutual Information.

Assuming that $\Pm_{XY}(E)\geq \Pm_X\Pm_Y(E)$ (typical scenario of interest), we can show the following:
\begin{Corollary}\label{HellBound}
    Let $E\in \mathcal{F}$ and let $\Pm_{XY}(E)\geq \Pm_X\Pm_Y(E)$, we have that:
    \begin{equation}
    \Pm_{XY}(E)-\Pm_X\Pm_Y(E) \leq H^2(X;Y)+ 2H(X;Y)\sqrt{\Pm_X\Pm_Y(E)}, \label{eq:hellSquareBound}
    \end{equation}
    where $H^2(X;Y)$ denotes $H^2(\Pm_{XY}\|\Pm_{X}\Pm_{Y}).$
\end{Corollary}
\begin{proof}
Let $\phi(t)=(\sqrt{t}-1)^2$. We have that $I_\phi(X,Y)=H^2(\Pm_{XY}\|\Pm_X\Pm_Y)=H^2(X;Y)$, \textit{\textit{i.e.}}, the squared Hellinger Distance of the joint from the product of the marginals.  Moreover, $\phi(t)$ is strictly increasing and invertible when restricted to $[1,+\infty)$ and $\phi^{-1}(t)=t+1+2\sqrt{t}$.
\eqref{eq:hellSquareBound} follows from Theorem~\ref{fDivBound} as stated in \eqref{fDivSimpleBound}. In particular: starting from Ineq.~\eqref{stepBeforeLastFDiv} we have that the inverse of $\phi$ is applied to an inequality of the form  $c \geq \phi\left(\frac{\Pm_{XY}(E)}{\Pm_{X}\Pm_{Y}(E)}\right)$. Given that $\frac{\Pm_{XY}(E)}{\Pm_{X}\Pm_{Y}(E)}\geq 1$ by assumption and using the invertibility of $\phi$ on $[+1,\infty)$ we recover \eqref{eq:hellSquareBound} after some algebraic manipulations.
\end{proof} 
When $X$ and $Y$ are independent, one has that $H(X;Y)=H^2(X;Y)=0$ and Corollary \ref{HellBound} recovers $\Pm_{XY}(E)=\Pm_X\Pm_Y(E)$. On the other hand, if $Y=X\sim \mathcal{U}([n])$ then, if $E=\{(x,y)\in [n]\times[n] | x=y\}$, \begin{equation}1=\Pm_{XY}(E)\leq 1 - \frac{1}{n^{3/2}} + 2 \sqrt{\left(1-\frac{1}{n^{3/2}}\right)\frac1n}.\end{equation} Thus, the bound is asymptotically tight even when $Y$ depends strongly on $X$. Regardless, the same reasoning that compared Maximal Leakage to $\chi^2$ applies: computing $H(X;Y)$ requires access to the marginal distributions $\Pm_{X},\Pm_{Y}$ and can be very complicated. Indeed, even for simple additive noise channels, no closed form expression is known for $H(X^n;Y)$ (or even for $TV(X^n;Y)$). In the context of learning instead, even for simple gradient descent mechanisms \cite[Example 2]{genErrWassDist2}, computing such measures can be very hard. It is, in general, possible to bound the divergence measures for every $\Pm_X$ (\textit{e.g.}, maximizing over all the possible $\Pm_X$), but this often implies using Maximal Leakage as a distribution-independent upper-bound on the chosen measure.
To conclude this section, let us then state the generalization error and sample complexity bounds provided by Hellinger distance.
\begin{Corollary}\label{HellDivBound}
Let $\mathcal{A}:\Z^n \to \mathcal{H}$ be a learning algorithm that, given a sequence $S$ of $n$ points, returns a hypothesis $h\in \mathcal{H}$. Suppose $S$ is sampled i.i.d according to some distribution $\mathcal{P}$ over $\Z$. Let $\ell:\mathcal{H}\times\Z 
	\to \mathbb{R}$ be a loss function such that $\ell(h,Z)$ is a $\sigma^2$-sub-Gaussian random variable, for some $\sigma$ and for every $h\in\mathcal{H}$.  
	Given $\eta \in (0,1)$, let $E=\{(S,h):|L_{\mathcal{P}}(h)-L_S(h)|>\eta \}$.
\begin{align}
     \mathbb{P}(E) &\leq 2\exp\left(-n\frac{\eta^2}{2\sigma^2}\right)+H^2(S;\mathcal{A}(S)) + 2^{3/2}H(S;\mathcal{A}(S))\exp\left(-n\frac{\eta^2}{4\sigma^2}\right)\\
     &\leq 2\exp\left(-n\frac{\eta^2}{2\sigma^2}\right)+H^2(S;\mathcal{A}(S)) + 2^{3/2}H(S;\mathcal{A}(S)).
\end{align}
	and in order to ensure a confidence of $\delta\in(0,1)$, \textit{i.e.} $\mathbb{P}(E)\leq \delta$, it is sufficient to have $m$ samples where
	\begin{equation}
        m\geq \frac{\log\left(\frac{1}{\delta-H^2(S;\A(S))+2^{3/2}H(S;\A(S))}\right)}{4\eta^2}.
    \end{equation}
	
\end{Corollary}

\section{Comparison of the results}

While the relationship among $I_\alpha$ for various $\alpha$'s is clear, a detailed and complete understanding of the relationship among all the $f-$divergences and $\alpha-$divergences is still lacking and many works are trying to address the issue (e.g., \cite{fDivIneq}). Restricting ourselves to $\chi^2$-like divergences, a summary of our current understanding is the following:
\begin{itemize}
    \item $I_\alpha(X,Y) = \min_{\Q_Y} D_\alpha(\Pm_{XY}\|\Pm_X\Q_Y) \leq D_\alpha(\Pm_{XY}\|\Pm_X\Pm_Y)$;
    \item $I_{\alpha_1}(X,Y) \leq I_{\alpha_2}(X,Y)$ if $\alpha_1 \leq \alpha_2$ \cite{verduAlpha};
    \item $\ml{X}{Y} \geq \log(\chi^2(X,Y)+1) \geq I_2(X,Y).$
\end{itemize}
All of this seems to hint, at least in this family of divergences, that the tightest sample-complexity bound would be given by $I_2$ (where both $1/\alpha=1/\gamma=1/2$). Taking $\alpha\in (1,2)$ will also consistently improve the dependency of the bound with respect to the information measure term while rendering the bound closer and closer to $1$ and worsening the dependency with respect to $\delta$ in the sample complexity, as $\gamma$ will tend to $+\infty$. The best trade-off between these two quantities remains an open problem. Moreover, the bounds involving other measures, like the Hellinger distance, can be fundamentally different and are not yet well understood. Considering Corollary \ref{HellDivBound}, while the dependency with respect to $\delta$ and $\eta$ seems to be the right one, the role played by the information measure is not as clear. Possibly, finding the optimal $\phi$ in Theorem \ref{fDivBound}, for a given behaviour of $\Pm_X\Pm_Y(E)$, could also shed some light on whether or not one should consider functions $\phi$ outside of the $\chi^2$-like family (polynomials). A Taylor expansion argument shows that most $\phi-$divergences are, in the end, $\chi^2$-like but, while those measures blow-up in some deterministic settings, others like Total Variation and Hellinger Distance do not. This different behaviour could be key in obtaining the tightest bound in the learning theory framework as well.



\section{Adaptive Data Analysis}\label{adaptiveDA}
Other than providing a nice generalization of the classical bounds for adaptive scenarios, maximal leakage can also be employed in adaptive data analysis. 
The model of adaptive composition we will be considering is identical to the setting  in~\cite{statValidity,genAdap,maxInfo} and defined as follows: 
\begin{Definition} \label{adaptiveComp}
	Let $\X$ be a set. Let $S$ be a random variable over $\X^n$.
	Let $(\mathcal{A}_1,\ldots,\mathcal{A}_m)$ be a sequence of algorithms such that $\forall i: 1\leq i \leq m\quad  \mathcal{A}_i  : \mathcal{X}^n \times \mathcal{Y}_1 \times \ldots \times  \mathcal{Y}_{i-1} \to \mathcal{Y}_i$.
	Denote with $Y_1 =\A_1(S), Y_2=\A_2(S,Y_1), \ldots, Y_m = \A_m(S,Y_1,\ldots,Y_{m-1})$.
	The adaptive composition of $(\mathcal{A}_1,\ldots,\mathcal{A}_m)$ is an algorithm that takes as an input $S$ and sequentially executes the algorithms $(\A_1,\ldots,\A_m)$ as described by the sequence $(Y_i, 1\leq 1 \leq m)$.
\end{Definition}

This level of generality allows us to formalize the behavior of a data analysts who, after viewing the previous outcomes of the analysis performed, decides what to do next. A potential analyst would execute a sequence of algorithms that are known to have a certain property (e.g. generalize well) when used without adaptivity. The question we would like to address is the following: is this property also maintained by the adaptive composition of the sequence?
The answer is not trivial as, for every $i$, the outcome of $\mathcal{A}_i$ depends both on $S$ and on the previous outputs,  that depend on the data themselves. However, when this property is guaranteed by some measure that composes adaptively itself (like differential privacy or, as we will show soon, maximal leakage) then it can be preserved.
Indeed, being robust to post-processing, Maximal Leakage allows us to retain the generalization guarantees it provides, regardless of how one may manipulate the outcome of the algorithm: 
\begin{Lemma}[Robustness to post-processing]\label{robustnessLeakage}
	Let $\mathcal{X}$ be the sample space and let $X$ be distributed over $\X$. Let $\Y$ and $\mathcal{Y'}$ be output spaces, and consider $\mathcal{A}:\mathcal{X}\to \mathcal{Y}$ and $\mathcal{B}:\mathcal{Y}\to\mathcal{Y'}$. Then,
	$\ml{X}{\B(\A(X))} \leq \ml{X}{\A(X)}$.
\end{Lemma}
The proof is a direct application of the data processing inequality for maximal leakage.
The useful implication of this result is as follows: in terms of maximal leakage, any generalization guarantees provided by $\A$ cannot be invalidated by further processing the output of $\A$. 
Regarding adaptive composition of two algorithms, we retrieve the following:
\begin{Lemma}[Adaptive Composition of Maximal Leakage]
	Let $\mathcal{A}:\mathcal{X}\to \mathcal{Y}$ be an algorithm such that $\Le(X\to \mathcal{A}(X)) \leq k_1$. 
	Let $\mathcal{B}:\mathcal{X}\times \mathcal{Y}\to \mathcal{Z}$ be an algorithm such that for all $ y \in \mathcal{Y},~ \Le(X\to\mathcal{B(}X,y))\leq k_2$. 
	Then $\Le \Big(X\to \big(\mathcal{A(}X),\mathcal{B}(X,\mathcal{A}(X))\big)\Big) \leq k_1+k_2$.
\end{Lemma}
\noindent The proof of this lemma relies crucially on the fact that maximal leakage depends on the marginal $\Pm_X$ only through its support and can be found in Appendix \ref{app:leakageProp}, along with the other proofs for this section.
In order to generalize the result to the adaptive composition of $n$ algorithms, we need to lift the property stated in the inequality \eqref{ineqLeak} to more than two random variables.
\begin{Lemma}\label{compositionN}
	Let $n\geq1$ and $X,A_1,\ldots,A_n$ be random variables.
	\begin{align} \ml{X}{(A_1,\ldots,A_n)} \leq \ml{X}{A_1} + \ml{X}{A_2|A_1}+\ldots  + \ml{X}{A_n| (A_1,\ldots,A_{n-1})}.\end{align}
\end{Lemma} \noindent 
The proof can be found in Appendix \ref{app:leakageProp}.
An immediate application of Lemma \ref{compositionN}  leads us to the following result.
\begin{Lemma}\label{thCompositionN}
	Consider a sequence of $k\geq 1$ algorithms: $(\mathcal{A}_1,\ldots,\mathcal{A}_k)$ where for each  $ 1\leq i \leq k$, $\mathcal{A}_i : \mathcal{X}\times \mathcal{Y}_1\times\ldots \times \mathcal{Y}_{i-1} \to \mathcal{Y}_i$. Suppose that for all $ 1\leq i \leq k$ and for all  $(y_1,\ldots,y_{k-1})\in \mathcal{Y}_1\times\ldots \times \mathcal{Y}_{i-1}$ , $\ml{X}{\mathcal{A}_i(X,y_1,\ldots,y_{i-1})} \leq j_i$. Then, denoting by $A_1,\ldots, A_n$ the (random) outputs of the algorithms:
	\begin{equation} \ml{X}{(A_1,\ldots,A_k)} = \ml{X}{ A^k} \leq \sum_{i=1}^n j_i. \end{equation}
\end{Lemma}

The conclusion to be drawn is straightforward: given a collection of algorithms that have bounded leakage (and thus good generalizations capabilities) even if the outcome of one of them is used to inform a subsequent analysis (hence, creating multiple dependencies on the data) the generalization guarantees of the composition can still be maintained.

Another interesting application of Corollary \ref{adaptML} in adaptive scenarios may be the following (same setting of \cite{maxInfo}): consider the problem of bounding the probability of making a false discovery, when the statistic to apply is selected with some data dependent algorithm $\mathcal{T}$. In this context, the classical guarantees that allow to upper-bound this probability by the significance value no longer hold. Measuring the information leaked from the data through $\mathcal{T}$ with the maximal leakage we retrieve the following:
\begin{Corollary}\label{hypTestML}
	Let $\mathcal{A} : \mathcal{X}^n \to \mathcal{T}$ be a data dependent algorithm
	for selecting a test statistic $t\in\mathcal{T}$. Let $X$ be a random dataset over $\mathcal{X}^n$. Suppose that $\sigma\in [0, 1]$ is the significance level chosen to control the false discovery probability. Denote
	with $E$ the event that $\mathcal{A}$ selects a statistic such that the null
	hypothesis is true but its p-value is at most $\sigma$. Then,
	\begin{equation}\mathbb{P}(E) \leq \exp(\ml{X}{ \mathcal{A}(X)}) \cdot \sigma.\end{equation}
\end{Corollary}
If the analyst wishes to achieve a bound of $\delta$ on the probability of making a false discovery in adaptive settings, the significance level $\sigma$ to be used should be no higher than $\delta/\exp(\ml{X}{\A(X)})$.
Once again, if $\mathcal{A}$ is independent from $X$, we recover the bound of $\sigma$.
\section{Comparison with other bounds}\label{comparison}
\begin{table*}
\begin{center}
\caption{Comparison between bounds}
\label{comparisonTable}
\begin{tabular}{c c c c c} 
\toprule
	 & Robust & Adaptive & Bound & Sample Complexity \\
	\midrule
	$\beta-$Stability \cite{bousqet}
	& No & No & exp. decay in $n$ & $f(\beta,\eta)\times\log\left(\frac{2}{\delta}\right)$ \\
	\addlinespace
	$\epsilon$-DP \cite{genAdap}
	& Yes & Yes &$\frac{1}{4} \exp{\left(\frac{-n\eta^2}{12}\right)}$, $\epsilon\leq \eta/2$ & $ \frac{12\cdot\log(1/4\delta)}{\eta^2}$\\
	\addlinespace
	MI \cite{learningMI}
	& Yes & Yes &$(I(S;Y) +1)/(2n\eta^2 -1)$ & $I(S;Y)/\eta^2\times 1/\delta$ \\
	\addlinespace
	Maximal Leakage
	& Yes & Yes & $2\cdot\exp(\ml{S}{Y}-2n\eta^2)$ & $\left(\ml{S}{Y} + \log\left(\frac{2}{\delta}\right)\right)/2\eta^2$ \\
	\addlinespace
	$\alpha$-Sibson's MI 
	& Yes & Unknown & $\exp(\frac{\alpha-1}{\alpha}(I_\alpha(S,Y))+\log2-2n\eta^2))$ & $\left(I_\alpha(S,Y)+\log2+\gamma\log\left(\frac{1}{\delta}\right)\right)/2\eta^2$ \\
	\addlinespace 
	$\chi^2$ & Yes & Unknown & $\sqrt{2}\exp\left(\frac12\left(\log(\chi^2(S,Y)+1)-2n\eta^2\right)\right)$ & 	$\left(\log(\chi^2(S,Y)+1)+2\log\left(\frac{\sqrt{2}}{\delta}\right)\right)/\eta^2$  \\
	\addlinespace
	VC-Dim. $d$ \cite{learningBook}
	& & & $2\cdot\exp(\log(K)-2n\eta^2)$ & 	$\left(d+\log\left(\frac{2}{\delta}\right)\right)/2\eta^2$ \\
	\bottomrule
\end{tabular}
\end{center}
\end{table*}
In this section, we compare the proposed new bounds to the existing ones from the literature. A summary is provided in Table \ref{comparisonTable}, where the bound involving Hellinger distance has been omitted as very different in shape (both in terms of the high-probability bound and, as a byproduct, in terms of the sample complexity bound). At a glance, from Table \ref{comparisonTable}, it is easy to see that most of the information measures bounds (with the sole exception of the Mutual Information one) can have an exponential decay with the number of samples $n$. This is indeed the desired behaviour with respect to $n$. The reason is that the event $E$ we consider is the event that the empirical average of some function (empirical risk) evaluated on a sequence of iid random variables (\textit{i.e.}, $S$, the training samples) diverges from its actual average (risk) more than some constant $\eta$. Even in the case where such function is independent from the samples $S$ the decay one typically finds in the literature is at best exponential with respect to $n$ (e.g., McDiarmid's and Hoeffding's inequality). More detailed comparisons will appear in the following subsections.
\subsection{Maximal Leakage and Mutual Information}
\label{sec:leakageAndMutI}
One interesting result in the field, that connects the generalization error with Mutual Information, under the same assumptions of Corollary~\ref{generrSibs2}, is the following (Theorem 8 of \cite{learningMI}): 
\begin{equation}\label{generrMI}
\mathbb{P}(E) \leq \frac{I(S;\mathcal{A}(S))+\log 2}{2n\eta^2-\log 2}.
\end{equation}

Let us compare this result with Corollary~\ref{generrML} in terms of sample complexity.
From Corollary~\ref{generrML}, it follows that using a sample size of \begin{equation}m\geq \left(\frac{ \ml{S}{\mathcal{A}(S)} +\log(2/\delta)}{2\eta^2}\right),\end{equation} yields a learner for $\mathcal{H}$ with accuracy $\eta$ and confidence $\delta$ and this, in turn, implies that \begin{equation}m_\mathcal{H}(\eta,\delta)= O\left(\frac{ \ml{S}{\mathcal{A}(S)}+\log(1/\delta)}{\eta^2}\right).\end{equation} Using the same reasoning with inequality~\eqref{generrMI}, we get :
\begin{equation}\label{MIBound}
    m\geq \left(\frac{ I(S;\mathcal{A}(S))+\log2+\delta\log2}{2\eta^2\delta}\right),
\end{equation}
and thus,
\begin{equation}m_\mathcal{H}(\eta,\delta)= O\left(\frac{I(S;\mathcal{A}(S))}{\eta^2}\cdot \frac{1}{\delta}\right).\end{equation} Since $\ml{X}{Y}\geq I(X;Y)$ \cite{leakage}, in the regime where the two measures behave similarly, the reduction in the sample complexity is exponential in $\delta$. This same reasoning can be applied to the sample complexity of $I_\alpha$ for a given $\alpha\in(1,+\infty]$. The exponential improvement in $\delta$ remains, although with a worse constant multiplying the $\log(1/\delta)$ term.
Another source of comparison can be found in Example \ref{tightness1} and \ref{tightness2}. 
Considering the same two scenarios, when $X$ is independent from $Y$, with the mutual information bound we retrieve:
\begin{equation}
    \Pm_{XY}(E)\leq \frac{1}{-\log(\max_y\Pm_X(E_y))} =  \frac{1}{-\log(\zeta)},
\end{equation} 
which is much weaker than the bound $\Pm_{XY}(E)\leq \zeta$ that can be obtained directly from Ineq. \eqref{MLBound}. When $X=Y\sim \mathcal{U}([n])$ and $\zeta=1/n$ we have that Ineq. \eqref{MIBound} recovers:
\begin{equation}
    1=\Pm_{XY}(E)\leq 1+\frac{1}{\log n},
\end{equation}
that is asymptotically tight, while with Ineq. \eqref{MLBound} we recover:
\begin{equation}
    1=\Pm_{XY}(E)\leq \frac1n\cdot n = 1,
\end{equation}
and thus, our bound is matched with an exact equality.
\subsection{Maximal Leakage and Differential Privacy}
\label{sec:leakageAndDP}
In this section we will compare our results with the generalization guarantees provided by differential privacy (DP). The definition of $\epsilon$-differentially privacy ($\epsilon$-DP) is the following:
\begin{Definition}
    Let $\mathcal{A}:\mathcal{X}^n\to\mathcal{Y}$ be a randomized algorithm. $\A$ is $\epsilon$-DP if for every $\mathcal{S}\subseteq\Y$ and every $x,y\in \X^n$ that differ only in one position:
    \begin{equation}
        \mathbb{P}(\A(x)\in \mathcal{S})\leq e^\epsilon \mathbb{P}(\A(y)\in \mathcal{S}).
    \end{equation}
\end{Definition}
A relationship with Maximal Leakage can be established:
\begin{Lemma}\label{MLandDP}
	Let $\mathcal{A}:\mathcal{X}^n\to\mathcal{Y}$ be an $\epsilon$-DP randomized algorithm, then \begin{equation}\label{eq:MLandDP}\ml{X}{\mathcal{A}(X)}\leq \epsilon\cdot n.\end{equation}
\end{Lemma} 
\begin{proof}
	Let $Y=\mathcal{A}(X)$ and assume, for simplicity, that $Y$ is a discrete random variable (the proof for continuous $Y$ follows very similar arguments). Fix some $\hat{\mathbf{x}}\in\mathcal{X}^n$, $\forall\,\mathbf{x}\in\mathcal{X}^n$ we have that $\mathbf{x}$ and $\hat{\mathbf{x}}$ differ in at most $n$ positions and, iteratively applying the definition of DP, we have that $\mathbb{P}(Y=y|X=\mathbf{x})\leq e^{\epsilon\cdot n} \mathbb{P}(Y=y|X=\hat{\mathbf{x}}).$ Thus:
	\begin{align}\ml{X}{Y} &= \log \sum_{y\in\mathcal{Y}} \max_{\mathbf{x}\in\mathcal{X}^n} \mathbb{P}(Y=y|X=\mathbf{x})\\ &\leq \log\sum_{y\in\mathcal{Y}} e^{\epsilon\cdot n}\mathbb{P}(Y=y|X=\hat{\mathbf{x}}) \\ &= n\cdot\epsilon \end{align}
\end{proof}

This suggests an immediate application of Corollary  \ref{generrML}.
Indeed, suppose $\mathcal{A}$ is an $\epsilon$-DP algorithm, then:
\begin{align} 
\exp(\ml{X}{Y}-2n\eta^2) &\leq \exp(\epsilon n -2n\eta^2) \\ &= \exp(-n(2\eta^2-\epsilon)).\label{eq:MLandDPandGen}
\end{align}
In order for the bound to be decreasing with $n$, we need $2\eta^2-\epsilon>0$ leading us to $\epsilon<2\cdot \eta^2$, where $\eta$ represents the accuracy of the generalization error and $\epsilon$ the privacy parameter. Thus, for fixed $\eta$, as long as the privacy parameter is smaller than $2\cdot \eta^2$, we have guaranteed generalization capabilities for $\mathcal{A}$ with an exponentially decreasing bound. For $\epsilon \leq \eta/2$, it is shown in~\cite[Theorem 9]{statValidity} that $\mathbb{P}(E) \leq 1/4 \exp{\left(-n\eta^2/12\right)}.$ It is easy to check that, for large enough $n$, our bound is tighter if $\epsilon \leq 23/12\eta^2$.\\
It is possible to see that enforcing differential privacy on some algorithm $\A$ induces generalization guarantees similar to those stated in Corollary \ref{adaptML}: suppose $\A$ is $\epsilon$-DP, with \begin{equation}\epsilon\leq \sqrt{\frac{\log(1/\beta)}{2n}},\end{equation} and let $\max_y \Pm_X(E_y)\leq \beta$ then \cite[Theorem 11]{statValidity}:
\begin{equation}
    \mathbb{P}(E)\leq 3\sqrt{\beta}. \label{generrDP}
\end{equation}
 The results we are providing are qualitatively different: we do not require the imposition of some (possibly very strong) privacy criteria on the algorithm but rather propose a way of estimating how the probabilities we are interested in change, by measuring the level of dependence through Maximal Leakage.
 Moreover, given an $\epsilon$-DP algorithm the bound obtained via Ineq. \eqref{eq:MLandDP} can be tighter for certain regimes of $\epsilon$. Indeed, let:  \begin{equation}\epsilon < \frac{\log{(3/\sqrt{\beta})}}{n}\leq \sqrt{\frac{\log(1/\beta)}{2n}}, \end{equation} using ~\eqref{generrDP} we get a \emph{fixed} bound of $3\sqrt{\beta}$, while with Corollary~\ref{adaptML} and Lemma~\ref{MLandDP} we obtain that:
\begin{equation}
\exp({\Le(X\to Y)})\cdot \beta < \exp{(\log{(3/\sqrt{\beta}))}}\cdot \ \beta = 3\sqrt{\beta}.
\end{equation}
Hence, whenever the privacy parameter is lower than $1/n\log{(3/\sqrt{\beta})}$ we are able to provide a better bound. Notice that Lemma~\ref{MLandDP} can be quite loose: using Lemma \ref{laplacianMech} it is possible to see that for classical mechanisms that imply $\epsilon$-DP, Maximal Leakage can be much lower that $\epsilon\cdot n$. Indeed, using the result proven in Lemma \ref{laplacianMech}, we can find such an example:
\begin{Corollary}
Let $g : \mathcal{X}^n \to \mathbb{R}$ be a function of sensitivity $1/n$ and let $N\sim Lap(1/n\epsilon)$ then the mechanism
$\mathcal{M}(x) = g(x) + N$ is $\epsilon-$DP.
Without loss of generality we have that $|g(x)|\leq 1$ (e.g. 0-1 loss) and thus:
\begin{equation}
\ml{X}{\mathcal{M}(X)} = \log(1+ \epsilon \cdot n) < \epsilon \cdot n.\end{equation}
\end{Corollary}
More importantly, the family of algorithms with bounded Maximal Leakage is not restricted to the differentially private ones. It is easy to see, for instance, that whenever there is a deterministic mapping and $\epsilon$-DP is enforced on it, $\epsilon \geq + \infty$. Trying to relax it to $(\epsilon,\delta)-$Differential Privacy does not help either, as one would need $\delta\geq 1$ rendering it practically useless. On the other hand, if the algorithm has a bounded range the Maximal Leakage from input to output is always bounded, since $\Le(X\to Y)\leq \min\{\log|\X|,\log|\Y|\}$. This simple observations allows us to immediately retrieve another result~\cite[Theorem 9]{genAdap}: $\mathbb{P}(E) \leq |\Y|\cdot \beta$, where $\beta$ is such that $\mathbb{P}(E_y)\leq \beta$ for every $y$. Indeed, given a a random variable $Y$ with bounded support, $\ml{X}{Y}\leq \log |\mathcal{Y}|$ and from Corollary \ref{adaptML} we have that: \begin{equation}
    \mathbb{P}(E)\leq \max_y\mathbb{P}(E_y)\exp\left(\ml{X}{Y}\right) \leq \beta \cdot |\mathcal{Y}|.
\end{equation}
This shows how Corollary \ref{adaptML} is more general than both Theorems 6 and 9 of~\cite{genAdap}.\\
To conclude the comparison let us now state Corollary \ref{generrML} with a general sensitivity $c$:
\begin{equation}\mathbb{P}(E)\leq 2\cdot \exp\bigg(\ml{X}{Y}- \frac{2\eta^2}{c^2n}\bigg).
\end{equation}  By contrast, \cite[Cor. 7]{genAdap} states that whenever an algorithm $\A:\X^n\to \Y$ outputs a function $f$ of sensitivity $c$ and is $\eta/(cn)-$DP then, denoting with $S$ a random variable distributed over $\X^n$ and with \begin{equation}E=\{(S,f): f(S)-\E(f)\geq \eta\},\end{equation} we have that: \begin{equation}\mathbb{P}(E)\leq 3\exp(-\eta^2/(c^2n)).\end{equation}
It is easy to see that we have a tighter bound whenever the accuracy $\eta> n\cdot c$.
\subsection{Sibson's Mutual Information, Maximal Leakage and Max Information}
\label{sec:leakageAndMI}
Another tool used in the line of work started by Dwork et al. \cite{genAdap,maxInfo} is the concept of max-information. The definition is the following:
\begin{Definition}\cite[Def. 10]{genAdap}
Let $X,Y$ be two random variables jointly distributed according to $\mathcal{P}_{XY}$ and with marginals $\Pm_X,\Pm_Y$. The max-information between $X$ and $Y$, is defined as follows:
\begin{equation} I_{\infty}^{M}(X,Y)= \log{\sup_{(x,y)\in\X\times\Y}\frac{\Pm_{XY}(\{(x,y)\})}{\Pm_X(\{x\})\Pm_Y(\{y\})}},\end{equation}
while, the $\beta-$approximate max-information is defined as:
\begin{equation}
I_{\infty}^{M,\beta}(X,Y)= \log\sup_{\mathcal{O}\subseteq \X\times\Y , \Pm_{XY}(\mathcal{O})>\beta}\frac{\Pm_{XY}(\mathcal{O})-\beta}{\Pm_X\Pm_Y(\mathcal{O})}.
\end{equation}
\end{Definition}
\begin{Remark}
Notice that we slightly changed the notation from \cite{genAdap} in order to avoid confusion. $I_\infty^M(X,Y)$ does not correspond to Sibsons's $I_\infty$ but it actually corresponds to R\'enyi's $D_\infty$, \textit{\textit{i.e.}},  $I_{\infty}^{M}(X,Y)= D_\infty(\Pm_{XY}\|\Pm_X\Pm_Y)$. 
\end{Remark}
One of the main reasons that led to the definition of approximate max-information is related to the generalization guarantees it provides, now recalled for convenience.
\begin{Lemma}\cite[Thm. 13]{genAdap}\label{GenMaxInfo}
	Let $X$ be a random dataset in $\X^n$ and let $\A:\X^n\to \Y$ be such that for some $\beta\geq 0$, $I_\infty^{M,\beta}(X,\A(X))=k$. Let $Y=\A(X)$ then, for any event $E\subseteq X^n\times \Y$:
	\begin{equation}
	\Pm_{XY}(E)\leq e^k \Pm_X\Pm_Y(E) + \beta. \end{equation}
\end{Lemma}
The result looks quite similar to Corollary \ref{generrML}, but the two measures, Max-Information and Maximal Leakage, although related, can be quite different. In this section we will analyze the connections and differences between the two measures underlining the corresponding implications.
\begin{Lemma}
	Let $\mathcal{A}:\mathcal{X}^n\to\mathcal{Y}$ be a randomized algorithm such that $I_\infty^M(X,\mathcal{A}(X))\leq k$. Then, $\ml{X}{\mathcal{A}(X)}\leq k.$
\end{Lemma}
\begin{proof}
Denote with $Y=\A(X)$. Having a bound of $k$ on the Max-Information of $\mathcal{A}$ means that for all $x \in\mathcal{X}^n,$ and $y \in \mathcal{Y}, \mathbb{P}(Y=y|X=x)\leq e^k \cdot \mathbb{P}(Y=y)$ and this implies that $\ml{X}{Y}\leq k.$
\end{proof}
More generally, we can say the following.
\begin{Lemma}
$I_\infty^M(X,Y) \geq D_\alpha(\Pm_{XY}\|\Pm_X\Pm_Y) \geq I_\alpha(X,Y)$ for every $\alpha\in[1,+\infty]$.
\end{Lemma}
\begin{proof}
We have that $I_\infty^M(X,Y) = D_\infty(\Pm_{XY}\|\Pm_X\Pm_Y) \geq \ml{X}{Y} \geq I_\alpha(X,Y)$ for any $\alpha\in [1,+\infty)$.
\end{proof}
With respect to $\beta$-approximate max-information instead, we can state the following.
\begin{Lemma}\label{bMIandML}
	Let $\mathcal{A}:\mathcal{X}^n\to\mathcal{Y}$ be a randomized algorithm. Let $X$ be a random variable distributed over $\mathcal{X}^n$ and let $Y=\mathcal{A}(X)$. Suppose $X,Y$ are discrete random variables and denote with $\Pm_{XY}$ the joint distribution and with $\Pm_X,\Pm_Y$ the corresponding marginals. For any $\beta \in (0,1)$ and $\alpha \in (1,+\infty]$ \begin{equation} I_{\infty}^{M,\beta}(X,\A(X)) \leq \frac{\alpha-1}{\alpha}I_\alpha(X,\A(X)) +\log\left({\frac{1}{\beta}}\right)
	.\end{equation}
\end{Lemma}
\begin{proof}
Fix any $\beta>0$. Using \cite[Lemma 18]{genAdap} we have that if \begin{equation}\Pm_{XY}\left(\left\{(x,y)\in \X\times\Y \bigg|\frac{\Pm_{XY}(\{x,y\})}{\Pm_X(\{x\})\Pm_Y(\{y\})}\geq e^k\right\}\right) \leq \beta,
\end{equation} then \begin{equation}
    I_\infty^{M,\beta}(X,Y) \leq k.
\end{equation}
Denote with $Y=\A(X)$. 
We have that 
\begin{align}
&\Pm_{XY}\left(\left\{(x,y)\in \X\times\Y \bigg|\frac{\Pm_{XY}(\{x,y\})}{\Pm_X(\{x\})\Pm_Y(\{y\})}\geq \frac{\exp\left(\frac{\alpha-1}{\alpha}I_\alpha(X,Y)\right)}{\beta}\right\}\right)
\leq \\ &\frac{\mathbb{E}_{\Pm_{XY}}\left[\frac{\Pm_{XY}(\{X,Y\})}{\Pm_X(\{X\})\Pm_Y(\{Y\})}\right]\cdot \beta}{\exp\left(\frac{\alpha-1}{\alpha}I_\alpha(X,Y)\right)} \leq \\ & \frac{\E_{\Pm_Y}\left[\left(\E_{\Pm_X}\left[\left(\frac{\Pm_{Y|X}(\{Y\})}{\Pm_Y(\{Y\})}\right)^\alpha\bigg|Y\right]\right)^{1/\alpha}\right]\cdot \beta}{\exp\left(\frac{\alpha-1}{\alpha}I_\alpha(X,Y)\right)}  = \beta.
\end{align} 
Hence, $I_{\infty}^{M,\beta}(X,\A(X)) \leq \log\left(\frac{\exp\left(\frac{\alpha-1}{\alpha}I_\alpha(X,Y)\right)}{\beta} \right) = \frac{\alpha-1}{\alpha}I_\alpha(X,Y) + \log\left(\frac{1}{\beta}\right).$

Taking the limit $\alpha\to\infty$ one also gets that $I_{\infty}^{M,\beta}(X,\A(X)) \leq \ml{X}{Y} +  \log\left(\frac{1}{\beta}\right).$
\end{proof}
The role played by $\beta$ can lead to undesirable behaviors of $\beta$-approximate max-information. 
The following example, indeed, shows how $\beta$-approximate max-information can be unbounded while, in the discrete case, the Maximal Leakage between two random variables is always bounded by the logarithm of the smallest cardinality.
\begin{Example}
Let us fix a $\beta \in (0,1).$
Suppose $X\sim\text{Ber}(2\beta)$. We have that $\ml{X}{X}=\log |\text{supp}(X)|=\log2$. For the $\beta-$approximate max-information we have:
$I_\infty^{M,\beta}(X,X) \geq \log ((2\beta-\beta)/\beta^2) = \log(1/\beta)$. It can thus be arbitrarily large.
\end{Example}
Another interesting characteristic of max-information is that, differently from differential privacy, it can be bounded even if we have deterministic algorithms: this observation is implied by the connection with what in the literature is known as \enquote{description length} of an algorithm, and synthesized in the following result \cite{genAdap}:
Let $\mathcal{A}:\mathcal{X}^n\to\mathcal{Y}$ be a randomized algorithm, for every $\beta>0$, \begin{equation} I_{\infty}^{M,\beta}(\mathcal{A},n)\leq \log{\bigg(\frac{|\mathcal{Y}|}{\beta}\bigg)}.\label{descrMI}\end{equation}
In contrast, with Sibson's $I_\alpha$ for every $\alpha \in[1,+\infty)$ we have that \begin{equation}\label{mlAndBDL}I_\alpha(X,\A(X))\leq \ml{X}{\A(X)}\leq \log(|\mathcal{Y}|).\end{equation}
Clearly, being $0<\beta$ typically very small in the key applications, the corresponding multiplicative factors in the bounds are $(|Y|/\beta)$ and $|\Y|$, and the difference between the two bounds can be substantial. It is also worth noticing that \eqref{descrMI} can be seen as a consequence of Lemma \ref{bMIandML} and \eqref{mlAndBDL}.
The difference between the two measures is not uniquely restricted to deterministic mechanisms. The following is a simple example of a randomized mapping where Maximal Leakage is smaller than $\beta$-approximate-max-information, for small $\beta$.
\begin{Example}
Consider $X\sim\text{Ber}(1/2)$ and a random variable $Y$ with support $\Y=\{0,1,e \}$. Consider also the following randomized mapping: $\Prob(Y=e|X=x) = \alpha$ and $\Prob(Y=x|X=x)=1-\alpha$. That is, $Y$ can be interpreted as passing $X$ through a binary erasure channel with erasure probability $\alpha$.
In this case, the Maximal Leakage is $\ml{X}{Y}=\log(2-\alpha)$ \cite{leakageLong}; while, for $\beta$-Approximate max-information one finds (after a series of computations) that:  
                         $I_\infty^{M,\beta}(X,Y)=\log(2\cdot\max\{(1-\alpha-\beta)/(1-\alpha),  (1-\beta)/(1+\alpha)\});$
It is easy to see how for a fixed $\alpha$ and for $\beta$ going to $0$, Approximate Max-Information approaches $\log2$ while Maximal Leakage is strictly smaller.
\end{Example}

\section{Conclusion}
Our aim was to bound the probability of an event $E$ under the joint distribution $\Pm_{XY}$ via information measures and the probability of the same event under the product of the marginals $\Pm_X\Pm_Y$. We started presenting bounds involving Luxemburg and Amemiya norms. We then particularised one of these results as a family of bounds characterized by four parameters $\alpha,\gamma,\alpha',\gamma'\geq1$ , constrained by the following equality $\frac1\alpha+\frac1\gamma=\frac1\alpha'+\frac1\gamma'$ (\textit{\textit{i.e.}}, H\"older's conjugates). We explicit and analyze the following choices of parameters:
\begin{itemize}
\item with $\alpha'=\alpha$ and $\gamma'=\gamma$ we retrieve a family of bounds involving the R\'enyi's divergence of order $\alpha$. A rewriting of this result allowed us to also recover a bound involving $p$-Hellinger divergences with $p\in(1,+\infty)$;
\item with $\alpha'\to1$ and consequently, $\gamma'\to\infty$ we retrieve a family of bounds involving Sibson's Mutual Information of order $\alpha$;
\item with $\alpha'\to1$, $\gamma'\to\infty$, $\alpha\to \infty$ and $\gamma\to 1$, we retrieve a bound involving Maximal Leakage;  
\end{itemize}
We also provided a family of bounds involving $f-$divergences where $f$ is an invertible convex function. We focused in particular on Maximal Leakage, since its semi-closed form and the dependence on $\Pm_X$ only through the support make it more amenable to analysis. Moreover, we show that the measure is robust under post-processing and composes adaptively. The robustness to post-processing is true for any information measure satisfying the data-processing inequality. However, since we currently lack a definition of conditional Sibson's MI or $f-$mutual information it is not possible, for the moment, verifying whether or not these other measures also compose adaptively.  Another interesting property of Maximal Leakage, instead, is that the bound it provides represents a possible generalization of the classical inequalities in adaptive mechanisms. The comparison with the other approaches showed how this measure is less strict than Differential Privacy and yet still provides strong generalization guarantees. We also showed how, in regimes where Mutual Information and Maximal Leakage behave similarly, the leakage bound provides an exponential improvement in the sample complexity. In general, one can also see that the sample complexity induced by $I_\alpha$ and Maximal Leakage is actually optimal with respect to $\delta$ in the realizable case. This also shows how information measures play a role similar to the VC-dimension, but tailored to the specific algorithm rather than the hypothesis class itself. Indeed, while the VC-dimension is a property of $\mathcal{H}$ only, information measures depend also on the samples and on (the distribution induced by) the algorithm. 
Some bounds on expected generalization error were also provided but, probably as an artifact of the analysis, they are generally worse (for finite samples $n$) than the ones that use Mutual Information \cite{explBiasMI,learningMI,tighteningMI}.


%

\appendices
\section{Proof of the Generalised Hölder's inequality}\label{sec:GenHoldIneqProof}
Let us recall the statement:\\ 
Let $\psi$ be an Orlicz fucntion and $\psi^\star$ denote its Legendre-Fenchel dual (\textit{i.e.}, $\psi^\star(x) = \sup_{\lambda} \lambda x - \psi(\lambda)$), then for every couple of random variable $U,V$:
    \begin{equation}\mathbb{E}[UV] \leq \lVert U\rVert_\psi \lVert V\rVert^A_{\psi^\star}.\label{genHoldIneqApp}\end{equation}
\begin{proof}
For every $\sigma,t >0$ we have that:
\begin{align}
    \mathbb{E}[UV] &= \mathbb{E}\left[\sigma\frac{U}{\sigma}\frac1t Vt\right] \\
    &\overset{\setlabel{youngsGeneralisedHolder}}{\leq}  \frac{\sigma}{t}\mathbb{E}\left[\psi\left(\frac{|U|}{\sigma}\right) +\psi^\star(|V|t)\right]
\end{align}
Where \reflabel{youngsGeneralisedHolder} follows from Young's inequality for convex functions. Choosing $\sigma = \lVert U\rVert_\psi$:
\begin{align}
   \mathbb{E}[UV] &\leq  \frac{\lVert U\rVert_\psi}{t}\mathbb{E}\left[\psi\left(\frac{|U|}{\lVert U\rVert_\psi}\right) +\psi^\star(|V|t)\right]
    \\ 
    &\overset{\setlabel{ineqGenHold}}{\leq} \lVert U\rVert_\psi \frac{1+\mathbb{E}\left[\psi^\star(|V|t)\right]}{t}. \label{lastStepGenHold}
\end{align}
\reflabel{ineqGenHold} follows from the definition of Luxemburg norm, \textit{i.e.}, $\mathbb{E}\left[\psi\left(|U|/\lVert U\rVert_\psi\right)\right]\leq 1$.
Taking the infimum with respect to $t$ in \eqref{lastStepGenHold} gives us \eqref{genHoldIneqApp} by definition of Amemiya norm.
\end{proof}
\section{Properties of Maximal Leakage}\label{app:leakageProp}
In this appendix we will provide proofs for the properties of Maximal Leakage. Let us start with the Adaptive Composition of the measure and let us recall the statement for reference:\\ 	Let $\mathcal{A}:\mathcal{X}\to \mathcal{Y}$ be an algorithm such that $\ml{X}{\A(X)} \leq k_1$. 
	Let $\mathcal{B}:\mathcal{X}\times \mathcal{Y}\to \mathcal{Z}$ be an algorithm such that for all $ y \in \mathcal{Y},~ \ml{X}{\mathcal{B}(X,y)}\leq k_2$. 
	
	\noindent Then $\ml{X}{(\mathcal{A(}X),\mathcal{B}(X,\mathcal{A}(X)))} \leq k_1+k_2$.
\noindent The proof of this lemma relies crucially on the fact that maximal leakage depends on the marginal $\Pm_X$ only through its support. 
\begin{proof}
    Let us denote with $R_X$ the support of a random variable $X$.
	If we consider the second constraint in our assumption and denoting with $Z_y = \mathcal{B}(X,y)$, we get:
	\begin{align}
	&\forall y \in \mathcal{Y}\, \ml{X}{Z_y} \leq k_2 \iff \\&\forall y \in \mathcal{Y} \sum_{z_y\in R_{Z_y}} \max_{x\in R_X}  \mathbb{P}(z_y|x) \leq \exp(k_2) \iff \\&\forall y \in \mathcal{Y} \sum_{z_y\in R_{Z|Y=y}} \max_{x\in R_X}  \mathbb{P}(z|x,y) \leq \exp(k_2) .\end{align}
	The last step holds, since every $y$ generates a family of conditional distributions $\mathbb{P}(z_y|x)$ through $\mathcal{B}$ and this probability is just $\mathbb{P}(z|x,y)$,  with $z=\mathcal{B}(x,y)$. 
	Using this observation in the conditional leakage of (\ref{ineqLeak}):
	\begin{align}
	\ml{X}{Z|Y} &= \log \max_{y\in R_{Y}} \sum_{z \in R_{Z|Y=y}} \max_{x \in R_{X|Y=y}} \mathbb{P}(z|x,y) 
	\\ &\leq \log \max_{y\in R_Y} \sum_{z \in R_{Z|Y=y}} \max_{x \in R_X} \mathbb{P}(z|x,y) 
	\\ &\leq \log \max_{y\in R_Y} \exp(k_2) \\ &= k_2,
	\end{align}
	leading us to the desired bound.
\end{proof}
Let us now show the generalization of this property to $n$ random variables. The statement reads:\\
	Let $n\geq1$ and $X,A_1,\ldots,A_n$ be random variables.
	\begin{align} \ml{X}{(A_1,\ldots,A_n)} \leq \ml{X}{A_1} + \ml{X}{A_2|A_1}+ \ldots  +\ml{X}{A_n| (A_1,\ldots,A_{n-1})}.\end{align}
\begin{proof}
	\begin{align}
	\ml{X}{(A_1, \ldots, A_n)} &= \ml{X}{A^n} \\ &= \ml{X}{(A^{n-1},A_n)},
	\end{align}
	then the result follows from recursively applying the same argument to $\ml{X}{A^{n-1}}$.
\end{proof}

\section{Examples}\label{app:ANC}
\subsection{Proof of Lemma~\ref{laplacianMech}}
We will now compute the value of the Maximal Leakage for an additive noise mechanism, where the noise is a Laplace random variable. Recall the statement of the lemma \ref{laplacianMech} is:\\
Let $g : \mathcal{X}^n \to \mathbb{R}$ be a function such that $g(x)\in [a,c], a<c ~\forall x\in\mathcal{X}^n$. The mechanism
$\mathcal{M}(x) = g(x) + N$ where $N \sim Lap(b)$ is such that:
\begin{equation}
    \ml{X}{\mathcal{M}(X)} = \log\left(1+\frac{(c-a)}{2b} \right)
\end{equation}

Let $Y=g(X)+N$, starting from Eq. \eqref{leakageContinous},
\begin{align}
     \exp(\ml{X}{Y}) &=  \int_{\mathbb{R}} \sup_{x: f_X(x)>0} f_{Y|X}(y|x) dy\\&= \int_{\mathbb{R}} \sup_{x: f_X(x)>0} f_{N}(y-g(x)) dy \\ &= \frac{1}{2b} \left( \int_{-\infty}^{+\infty} \sup_{x: \Pm_X(x)>0} \exp{\left(\frac{-|y-g(x)|}{b}\right) dy}  \right) \\
   &=  \frac{1}{2b} \left( \int_{-\infty}^{a} \exp{\left(\frac{-|y-a|}{b}\right) dy}  + \int_{a}^{c} dy\right) \\ &+ \frac{1}{2b}\left(\int_{c}^{+\infty} \exp{\left(\frac{-|y-c|}{b}\right) dy} \right) \\ &=\frac{1}{2b} \left( \int_{-\infty}^{0} \exp{\left(\frac{-|z|}{b}\right) dz}  + (c-a)\right) \\& +\frac{1}{2b}\left(\int_{0}^{+\infty} \exp{\left(\frac{-|w|}{b}\right) dw} \right) \\ &= \frac{1}{2b} \left(  (c-a) + 2 \int_{0}^{+\infty} \exp{\left(\frac{-w}{b}\right) dw} \right) \\ &= \frac{1}{2b}(  (c-a) + 2b) = \left( 1 + \frac{(c-a)}{2b}  \right). 
\end{align}
The proofs of the other additive noise mechanisms (Gaussian and Exponential) follow along the same lines as the proof of Lemma~\ref{laplacianMech}. $\hfill \blacksquare$ 

\subsection{Proof of Corollary~\ref{noisyERM}}
Suppose the hypothesis space is countable and let $k:=|\mathcal{H}|$ (could be infinite). Suppose also that $\mathbb{E}[N_i]=b_i$ \cite{infoThGenAn} (with $N_i$ being the noise added to the $i$-th hypothesis). Since the choice of the hypothesis depends only on the noisy empirical errors, the following is a Markov Chain $S-(L_S(h_i))_{i\in[k]}-(L_S(h_i) + N_i)_{i\in[k]}-H$. Then by the data-processing inequality for Maximal Leakage:
\begin{align}
\ml{S}{H} &\leq \ml{(L_S(h_i))_{i\in[k]}}{(L_S(h_i) + N_i)_{i\in[k]}}.
\end{align}
Also, denoting with $X_i= L_S(h_i)$ and with $Y_i=X_i+N_i$:
\begin{align}
    \exp(\ml{(X_1,\ldots,X_k)}{(Y_1,\ldots,Y_k)}
     &= \idotsint_{-\infty}^{+\infty} \max_{x^n} f(y^n|x^n) dy^n \\&= \idotsint_{-\infty}^{+\infty} \max_{x^n} \left(\prod_{i=1}^k f_{N_i}(y_i-x_i)\right) dy^n \\ 
     &=\idotsint_{-\infty}^{+\infty} \max_{x^n}\left(\prod_{i=1}^k \frac{1}{b_i} e^{-(y_i-x_i)/b_i}\right) dy^n \\ 
     &= \prod_{i=1}^k \int_{-\infty}^{+\infty} \max_{x_i} \left(\frac{1}{b_i} e^{-(y_i-x_i)/b_i}\right) dy \\ 
     &= \prod_{i=1}^k \left( 1+\frac{1}{b_i} \right).\label{mlExpNoise}
\end{align}
Equation \eqref{mlExpNoise}, along with Corollary \ref{generrML}, implies that:
\begin{align}
    \mathbb{P}(\text{gen-err}(\mathcal{A})\geq \eta) &\leq2\exp\left( \ml{S}{H} - 2n\eta^2\right) \\  
    &=  2 \exp\left( \sum_{i=1}^k \log\left( 1+ \frac{1}{b_i} \right) - 2n\eta^2\right).
\end{align}
Now, suppose that $b_i=i^{1.1}/n^{1/3}$, 
  \begin{align}
  \ml{S}{H}&\leq \sum_{i=1}^k \log(1+ n^{1/3}/i^{1.1})\\ &\leq n^{1/3}\sum_{i=1}^{+\infty} \frac{1}{i^{1.1}}\\ &\leq (n^{1/3})\cdot  11.
  \end{align} We have that  
  \begin{equation}
  \mathbb{P}(\text{gen-err} \geq \eta )\leq 2\exp(-n(2\eta^2-11/n^{2/3})).
  \end{equation} 

\section{Expected Generalization Error}\label{sec:expectedGenErr}
 Given the generalization error bounds proposed so far, one may ask how these reflect in results on the expected value of the generalization-error. To give a meaningful bound one needs to make some assumptions on the probability of our event $E$, in particular we will assume this probability to be exponentially decreasing with the number of samples $n$ (as it often happens in the literature \cite{bousqet,BLM2013Concentration}). This section will focus only on Sibson's $\alpha$-Mutual Information. It is possible to extend these results also to $f$-Mutual Information. However, in order to do so, one needs more information on $f$. For example, using the same techniques and starting from Corollary \ref{hellingerDivBound}, one can state a bound on the expected generalization error for $p-$Hellinger divergences. However, it is unclear how to derive a result involving all increasing and convex functions $f$. 
The following result is inspired by \cite[p. 419]{learningBook} with a different (slightly improved, for our purposes) proof.
\begin{Lemma}\label{lemmaBoundExp}
	Let $X$ be a random variable and let $\hat{x}\in\mathbb{R}$. Suppose that there exist $a\geq 0$ and $b\geq e$ such that for every $\eta>0$ $\Pm_X(|X-\hat{x}|\geq \eta)\leq 2b\exp\left(-\eta^2/a^2\right)$ then $\mathbb{E}\left[|X-\hat{x}|\right]\leq a\min\left\{3\sqrt{\log b},2\sqrt{\log 2b}\right\}$.
\end{Lemma}
\begin{proof}
	Since $|X-\hat{x}|$ is a positive random variable we have that \begin{equation}\label{exp1}\mathbb{E}\left[|X-\hat{x}|\right]=\int_0^{+\infty} \Pm_X(|X-\hat{x}|\geq \eta) d\eta.\end{equation} Since for small values of $\eta$ the exponential bound may be exceedingly loose, instead of trivially upper-bounding \eqref{exp1} we do the following: 
	\begin{align}
	\mathbb{E}\left[|X-\hat{x}|\right]=&\int_0^{+\infty} \Pm_X(|X-\hat{x}|\geq \eta)d\eta \\
	\leq &\int_0^{+\infty} \min\left(1,2b\exp\left(-\eta^2/a^2\right)\right) d\eta \label{multFactor} \\ 
	= &\int_0^{\sqrt{a^2\log 2b}} d\eta + \int_{\sqrt{a^2\log 2b}}^{+\infty} 2b\exp(-\eta^2/a^2)d\eta  \\
	\leq &\sqrt{a^2\log 2b} + \frac{a^2}{\sqrt{a^2\log 2b}}\int_{\sqrt{a^2\log 2b}}^{+\infty} \frac{2b\eta}{a^2}\exp(-\eta^2/a^2)d\eta \\
	= &a\left(\sqrt{\log 2b}+ \frac{1}{\sqrt{\log 2b}}\right) \\
	\leq &a\min\left\{3\sqrt{\log b},2\sqrt{\log 2b}\right\}.
	\end{align}
\end{proof}
\begin{Theorem}\label{boundExp}
	Let $\mathcal{A}:\mathcal{Z}^n\to \mathcal{H}$ be a learning algorithm and let $I_\alpha(S,\A(S))$ (\textit{i.e.}, Sibson's Mutual Information of order $\alpha$) be the dependence measure chosen. Suppose that the loss function $\ell:\mathcal{Z}\times\mathcal{H}\to \mathbb{R}$ is such that $\forall h, \Pm_{S\sim \mathcal{D}^n}(|L_S(h)-\mathbb{E}[L(h)]|>\eta)\leq 2\exp\left(-\frac{\eta^2}{2\sigma^2}n\right) $ for some $\sigma>0$ (e.g. $\ell(h,Z)-{\mathbb{E}}[\ell(h,Z)]$ is $\sigma^2$-sub-Gaussian for each $h$), then:
	\begin{align}
	&\mathbb{E}\left[|L_S(H)-\mathbb{E}[L(H)]|\right] \leq \sqrt{\frac{8\sigma^2(\log(2)+I_\alpha(S,\A(S)))}{n}}.
	\end{align}
\end{Theorem}
\begin{proof}
	The proof is a simple application of Lemma \ref{lemmaBoundExp} and Corollary \ref{generrSibs2} with $a=\sqrt{2\gamma\sigma^2}/\sqrt{n}$ and
	
	$b=2^{\frac1\gamma-1}\exp\left(\frac1\gamma I_\alpha(S,\mathcal{A}(S))\right)$.
\end{proof}
\begin{Remark}
    Notice that, even though we provide a concrete example (Theorem \ref{boundExp}) that uses $\sigma^2$ sub-Gaussianity the assumption is not strictly necessary.
    Lemma \ref{lemmaBoundExp} only requires that the probability of $X$ diverging from $\hat{x}$ decays exponentially fast. This can be true also for other classes of random variables, like sub-Weibull ones with an opportune choice of parameters \cite{subWeibull}. Moreover, unlike \cite{learningMI,explBiasMI}, Corollary \ref{alphaDivBound} and \ref{sibsMIBoundCor} are more general and do not require any assumption about the convergence rate.
\end{Remark}
An important result, obtained through a different route, is the bound on the expected generalization error via Mutual Information (Theorem 1 of \cite{infoThGenAn}). We restate it here for ease of reference. Under the assumption that $\ell(h,Z)$ is $\sigma^2$-sub Gaussian for each $h$:
\begin{equation}
    |\mathbb{E}\left[L_S(H)-\mathbb{E}[L(H)]\right]| \leq \sqrt{\frac{2\sigma^2}{n} I(S;\A(S))}. \label{mutualExpectedGenErr}
\end{equation}
In the spirit of comparison, let us also state a similar bound using Theorem \ref{boundExp} but with $\alpha\to\infty$ and using $3a\sqrt{\log b}$ as a bound on the expected value. Setting $a=\sqrt{\frac{2\sigma^2}{n}},\, b=\exp(\ml{S}{\A(S)})$ one retrieves the following:
\begin{equation}
    |\mathbb{E}\left[L_S(H)-\mathbb{E}[L(H)]\right]| \leq \mathbb{E}\left[|L_S(H)-\mathbb{E}[L(H)]|\right] \leq 3\sqrt{\frac{2\sigma^2}{n}\ml{S}{\A(S)}}.\label{leakageExpectedGenErr}
\end{equation}
We have seen before (c.f., Section \ref{sec:leakageAndMutI}) that Sibsons's $\alpha$-MI brings an exponential improvement in the dependency over $\delta$ when considering the bound on the large deviation event. However, the measure does not seem to bring any improvement in controlling the expected generalization error. Looking at Ineq. \eqref{leakageExpectedGenErr} one immediately sees that (other than being a constant away from Ineq. \eqref{mutualExpectedGenErr}) Ineq. \eqref{leakageExpectedGenErr} is always going to provide a looser bound, as $I_\alpha(S,\A(S))\geq I(S;\A(S))$ for every $\alpha>1$. This represents, perhaps, an artifact of the analysis used to prove Lemma \ref{lemmaBoundExp}.

\section{Alternative Proofs}\label{sec:altProofsApp}

\subsection{Alternative Proofs of Theorem \ref{fDivBound}}\label{sec:altProofTheoremFDiv}
We have here two alternative proofs for Theorem \ref{fDivBound}, whose statement says:\\
 Let $\phi : \mathbb{R}\to \mathbb{R}$ be a convex function such that $\phi(1)=0$, and assume $\phi$ is non-decreasing on $[0,+\infty)$. Suppose also that $\phi$ is such that for every $y\in\mathbb{R}^+$ the set $\{t \geq 0 : \phi(t)>y \}$ is non-empty, \textit{i.e.} the generalized inverse, defined as $\phi^{-1}(y)= \inf\{t \geq 0 : \phi(t) > y\}$,  exists. Let $\phi^\star(t)= \sup_{\lambda\geq 0} \lambda t -\phi(\lambda)$ be the  Fenchel-Legendre dual of $\phi(t)$ \cite[Section 2.2]{BLM2013Concentration}. Given an event $E\in\F$, we have that:
	\begin{align}
	\Pm_{XY}(E) \leq &\Pm_X\Pm_Y(E)\cdot \phi^{-1}\left(\frac{I_\phi(X,Y)+(1-\Pm_X\Pm_Y(E))\phi^\star(0)}{\Pm_X\Pm_Y(E)}\right). \label{generalBoundInv2App}
	\end{align}
\begin{proof}[Proof 2 of Theorem \ref{fDivBound}]
    Let $\phi:\mathbb{R}\to\mathbb{R}$ be a convex function respecting all the assumption of Theorem \ref{fDivBound}. We have that, given two measures $\Pm,\Q$ \cite{varReprfDiv}: 
    \begin{equation}
        D_\phi(\Pm\|\Q) \geq \E_\Pm[f] - \E_\Q[\phi^\star(f)],\label{variationfDiv}
    \end{equation}
    for every $\Pm-$measurable function $f$.
    Let $\Pm=\Pm_{XY},\Q=\Pm_X\Pm_Y$ and $f(x,y)=\lambda\mathbbm{1}_{\{(x,y)\in E\}}$ for $\lambda>0$. By Inequality \eqref{variationfDiv} we have that:
    \begin{align}
        I_\phi(X,Y)&=D_\phi(\Pm_{XY}\|\Pm_X\Pm_Y) \\&\geq \E_{\Pm_{XY}}[\mathbbm{1}_E] - \E_{\Pm_X\Pm_Y}[\phi^\star(\mathbbm{1}_E)]\\
        &= \lambda\Pm_{XY}(E)- \E_{\Pm_X\Pm_Y}[\phi^\star(\mathbbm{1}_E)].
    \end{align}
    We thus have that:
    \begin{align}
        \Pm_{XY}(E)\leq \frac{I_\phi(X,Y)+\E_{\Pm_X\Pm_Y}[\phi^\star(\mathbbm{1}_E)]}{\lambda}.
    \end{align}
    The conclusion can be drawn by following the same steps of the other proof starting from Inequality \eqref{fInfo}.
\end{proof}
\begin{proof}[Proof 3 of Theorem \ref{fDivBound}]
	 $\forall \lambda>0$:
	\begin{align}
	 \Pm_{XY}(E) &=\E_{\Pm_{XY}}[\mathbbm{1}_E]  \\ &= \E_{\Pm_X\Pm_Y}\left[\mathbbm{1}_E  \frac{d\Pm_{XY}}{d\Pm_X\Pm_Y}\right]  \\ 
	& \stackrel{\text{(a)}}{\leq} \frac1\lambda\E_{\Pm_X\Pm_Y}\left[\phi^\star(\lambda \mathbbm{1}_E)+ \phi\left(\frac{d\Pm_{XY}}{d\Pm_X\Pm_Y}\right)\right] \label{youngIneq} 
	\\ & \stackrel{\text{(b)}} {\leq} \frac{I_\phi(X,Y) + \E_{\Pm_X\Pm_Y}[\phi^\star(\lambda\mathbbm{1}_E)]}{\lambda} \label{fInfo} \\
	& \stackrel{\text{(c)}}{\leq} \frac{I_\phi(X,Y) + \phi^\star(\lambda)\E_{\Pm_X\Pm_Y}[\mathbbm{1}_E]+\phi^\star(0)\E_{\Pm_X\Pm_Y}[1-\mathbbm{1}_E]}{\lambda} \label{convexity} \\ &= \frac{I_\phi(X,Y) + \phi^\star(\lambda)\Pm_X\Pm_Y(E)+\phi^\star(0)(1-\Pm_X\Pm_Y(E))}{\lambda}, \label{beforeMin}
	\end{align}
	where (a) follows from Young's inequality and where $\phi^\star$ is the Legendre-Fenchel dual of $\phi$, (b) follows from our definition of $\phi-$Mutual Information, and (c) follows from the fact that $\mathbbm{1}_E\in [0,1]$ and:
	\begin{align}
	\phi^\star(\lambda \mathbbm{1}_E) &= \label{beginConvexIneq} 
	\phi^\star(\lambda(\mathbbm{1}_E+(1-\mathbbm{1}_E)0)) \\ 
	&\leq \mathbbm{1}_E\phi^\star(\lambda)+(1-\mathbbm{1}_E)\phi^\star(0). \label{endConvexIneq}  
	\end{align} 
	To get the best bound over $\Pm_{XY}$ we can minimize \eqref{beforeMin} over all $\lambda>0$:
	\begin{align}
	\Pm_{XY}(E) &\leq \inf_{\lambda>0} \frac{I_\phi(X,Y) + \phi^\star(\lambda)\Pm_X\Pm_Y(E)+(1-\Pm_X\Pm_Y(E))\phi^\star(0)}{\lambda}  \\ &= \Pm_X\Pm_Y(E)\cdot \inf_{\lambda>0} \frac{\frac{I_\phi(X,Y)+(1-\Pm_X\Pm_Y(E))\phi^\star(0)}{\Pm_X\Pm_Y(E)} + \phi^\star(\lambda)}{\lambda} \\ &\stackrel{\text{(a)}} {=} \Pm_X\Pm_Y(E)\cdot {\left(\phi^{\star\star}\right)}^{-1}\left(\frac{I_\phi(X,Y)+(1-\Pm_X\Pm_Y(E))\phi^\star(0)}{\Pm_X\Pm_Y(E)}\right), \label{lastStepFDivProof}
	\end{align}
	where (a) follows from~\cite[Lemma 2.4]{BLM2013Concentration}. There is a slight difference in the assumptions (which do not affect the proof). In the notation of Lemma 2.4, set $\psi = \phi^\star$; since $\phi$ is convex and non-decreasing by assumption, then $\phi^\star$ is also convex and non-decreasing, as required for $\psi$ in the Lemma. Similarly, $\psi^\star$ in the Lemma is convex and non-decreasing, which is true for our $\phi$ by assumption (note that $\psi(0)=\psi^\star(0)=0$ in the Lemma are technicalities that do not affect the required equality). 

Finally, the simplified form follows from upper bounding~\eqref{endConvexIneq} by $\mathbbm{1}_E\phi^\star(\lambda)$. It can also be seen directly from the general formula by noting that $\phi^{-1}$ is non-decreasing (which follows from its definition and the fact that $\phi$ is non-decreasing).
If we also assume that $\phi^\star(t) := \sup_{\lambda\geq 0}(\lambda t - \phi(t)) = \sup_{\lambda\in\mathbb{R}}(\lambda t - \phi(t)) $ we have that ${\phi^{\star\star}}^{-1} = \phi^{-1}$ and we recover \eqref{generalBoundInv2App}.
\end{proof}
\subsection{Alternative Proof of Corollary \ref{adaptML}}\label{sec:altProofMaximalLeakage}
To conclude, let us provide an alternative proof of Corollary \ref{adaptML} that goes through $D_\infty$. The statement is:\\
Let $E\in\F$ we have that:
\begin{equation}
\Pm_{XY}(E) \leq \left(\esssup_{\Pm_y} \Pm_X(E_y)\right) \exp\left(\ml{X}{Y}\right). 
\end{equation}
\begin{proof}[Proof 2 of Corollary \ref{adaptML}]
Let $\Pm,\Q$ be two measure over the same $\sigma$-field we have that the $\alpha$-divergence of order infinity $D_\infty(\Pm\|\Q)= \log \left(\esssup_{d\Q} \frac{d\Pm}{d\Q}\right)$\cite{alphaDiv}. Suppose that the measurable space $(\X\times\Y,\F)$  has the regular conditional probability property \cite{regularCondProb} and let us denote the conditional measure of $\Pm_{XY}$ with respect to the random variables $X,Y$ as $\Pm_{X|Y=y}$ and $\Pm_{Y|X=x}$. We have that, for every $y\in\Y$:
\begin{align}
    \Pm_{X|Y=y}(E_y) &\leq \esssup_{\Pm_y}\Pm_X(E_y)\cdot \exp\left(D_\infty\left(\Pm_{X|Y=y}\|\Pm_X\right)\right)\\
    &= \esssup_{\Pm_y}\Pm_X(E_y)\cdot \exp\left(D_\infty\left(\Pm_{Y|X=x}\|\Pm_Y\right)\right).
\end{align}
And thus:
\begin{align}
    \Pm_{XY}(E) &= \mathbb{E}_{\Pm_Y}\left[\Pm_{X|Y=y}(E_y)\right]  \\
    &\leq \esssup_{\Pm_y}\Pm_X(E_y)\mathbb{E}_{\Pm_Y}\left[\exp\left(D_\infty\left(\Pm_{Y|X=x}\|\Pm_Y\right)\right)\right] \\
    &= \esssup_{\Pm_y}\Pm_X(E_y)\cdot \exp\left(\ml{X}{Y}\right),
\end{align}
as by \cite[Thm. 7]{leakageLong}:
\begin{equation}
    \mathbb{E}_{\Pm_Y}\left[\exp\left(D_\infty\left(\Pm_{Y|X=x}\|\Pm_Y\right)\right)\right] = \mathbb{E}_{\Pm_Y}\left[\esssup \frac{d\Pm_{X|Y}}{d\Pm_Y}\right]
\end{equation}
\end{proof}
\subsection{Alternative Proofs of Corollary \ref{alphaDivBound}}\label{sec:altProofsCorollaryAlphaDiv}
We will now provide alternative proofs for Corollary \ref{alphaDivBound}. The statement reads:\\
Let $E\in\F$ we have that:
\begin{equation}
\Pm_{XY}(E)\leq (\Pm_X\Pm_Y(E))^{1/\gamma}\exp\left(\frac{\alpha-1}{\alpha}D_\alpha(\Pm_{XY}\|\Pm_X\Pm_Y)\right).
\end{equation}
\begin{proof}[Proof 2 of Corollary \ref{alphaDivBound}]
Let us denote with $p = \Pm_{XY}(E),q =\Pm_X\Pm_Y(E), \bar{p} = 1-p, \bar{q}=1-q $
\begin{align}
    D_\alpha(\Pm_{XY}\|\Pm_X\Pm_Y) &\overset{\setlabel{dataProcAlpha}}{\geq} D_\alpha(\text{Ber}(p)\|\text{Ber}(q))
    \\ &= \frac{1}{\alpha-1}\log\left(p^\alpha q^{1-\alpha}  + \bar{p}^\alpha\bar{q}^{1-\alpha}\right)
    \\ &\geq \frac{1}{\alpha-1}\log p^\alpha q^{1-\alpha},
\end{align}
where \reflabel{dataProcAlpha} follows from the Data-Processing inequality for $\alpha-$Divergences.
Re-arranging the terms one gets:
\begin{align}
   p^\alpha q^{1-\alpha} &\leq \exp\left((\alpha-1) D_\alpha(\Pm_{XY}\|\Pm_X\Pm_Y)\right) \iff 
   \\ p^\alpha &\leq \exp\left((\alpha-1) D_\alpha(\Pm_{XY}\|\Pm_X\Pm_Y)\right)q^{\alpha-1} \iff \\ p &\leq \exp\left(\frac{(\alpha-1)}{\alpha} D_\alpha(\Pm_{XY}\|\Pm_X\Pm_Y)\right)q^{\frac{\alpha-1}{\alpha}}.
\end{align}
\end{proof}

\begin{proof}[Proof 3 of Corollary \ref{alphaDivBound}]
Fix $\alpha>1$ and consider the following convex function:
\begin{equation}\phi_\alpha(t)=\frac{t^\alpha-1}{\alpha-1},
\end{equation}  \textit{i.e.} the Hellinger Divergence. The restriction of $\phi_\alpha(t)$ to $[0,+\infty)$ is increasing and thus invertible. Since we will consider only ratios between measures, the restriction is sufficient and Theorem \ref{fDivBound} is applicable. 
It follows that: \begin{equation}\phi_\alpha^{-1}(t)=\left(\left(\alpha-1\right)t+1\right)^{1/\alpha},\end{equation} and that:
\begin{equation}
    \phi_\alpha^\star(t)= t\left(\frac{(\alpha-1)t}{\alpha}\right)^{1/\alpha-1}-\frac{\left(\frac{(\alpha-1)t}{\alpha}\right)^{\alpha/\alpha-1}}{\alpha-1} + \frac{1}{\alpha-1},
\end{equation}
from which we can deduce that:
\begin{equation}\phi_\alpha^\star(0)=\frac{1}{\alpha-1}.\end{equation} We also have that for a given $\alpha>0$ and two measures $\Pm,\Q$ \cite{fDiv1}: 
\begin{equation}
D_\alpha(\Pm\|\Q)= \frac{1}{\alpha-1} \log(1+(\alpha-1)D_{f_\alpha}(\Pm\|\Q)) \label{divergEquality},
\end{equation}
then, with $\phi$=$f_\alpha$ and computing the right-hand side of Ineq. \eqref{generalBoundInv} we retrieve:
\begin{align}
&\phi^{-1}\left(\frac{I_\phi(X,Y)+ (1-\Pm_X\Pm_Y(E))/(\alpha-1)}{\Pm_X\Pm_Y(E)}\right)\\ &= \phi^{-1}\left(\frac{D_\phi(\Pm_{XY}\|\Pm_X\Pm_Y)+ (1-\Pm_X\Pm_Y(E))/(\alpha-1)}{\Pm_X\Pm_Y(E)}\right)\\ &= \left(\frac{(\alpha-1)D_{f_\alpha}(\Pm_{XY}\|\Pm_X\Pm_Y)+ (1-\Pm_X\Pm_Y(E))}{\Pm_X\Pm_Y(E)}+1\right)^{1/\alpha}
\\ 
&=  \left( \frac{(\alpha-1)D_{f_\alpha}(\Pm_{XY}\|\Pm_X\Pm_Y) + 1)}{\Pm_X\Pm_Y(E)}\right) ^{1/\alpha} \\ &\overset{\setlabel{divergEqInf}}{=} \frac{\exp\left(\frac{\alpha-1}{\alpha}D_\alpha(\Pm_{XY}\|\Pm_X\Pm_Y)\right)}{\Pm_X\Pm_Y(E)^{1/\alpha}}  \label{divergEqInf},
\end{align}
where \reflabel{divergEqInf} follows from \eqref{divergEquality}. To conclude, substitute \eqref{divergEqInf} in \eqref{generalBoundInv}:
\begin{align}
\Pm_{XY}(E) &\leq \Pm_X\Pm_Y(E)^\frac{\alpha-1}{\alpha}\cdot \exp\left(\frac{\alpha-1}{\alpha}D_\alpha(\Pm_{XY}\|\Pm_X\Pm_Y)\right).
\end{align}
since $\frac{\alpha-1}{\alpha}=\frac{1}{\gamma}$ is the H\"older's conjugate of $\frac1\alpha$ we recover Corollary \ref{alphaDivBound}.
\end{proof}
\begin{proof}[Proof 4 of Corollary 6]
Consider Theorem \ref{luxemburgNormBound} and choose $\psi(x) = \frac{x^\alpha}{\alpha}$ with $\alpha\geq 1$. We have that the Luxemburg norm of $\frac{d\Pm_{XY}}{d\Pm_X\Pm_Y}$, when the expectation is considered with respect to $\Pm_X\Pm_Y$ is
\begin{equation}
    \left\lVert \frac{d\Pm_{XY}}{d\Pm_X\Pm_Y} \right\rVert_\psi = \mathbb{E}_{\Pm_X\Pm_Y}^{1/\alpha}\left[\left(\frac{d\Pm_{XY}}{d\Pm_X\Pm_Y}\right)^{\alpha}\right]\cdot \left(\frac1\alpha \right)^{1/\alpha}.
\end{equation}
Moreover, $\psi^{-1}(x)=(\alpha x)^{1/\alpha}$.
Hence:
\begin{align}
    \Pm_{XY}(E)&\leq \Pm_X\Pm_Y(E)\psi^{-1}\left(\frac{1}{ \Pm_X\Pm_Y(E)}\right)\left\lVert \frac{d\Pm_{XY}}{d\Pm_X\Pm_Y} \right\rVert_\psi\\ &= \Pm_X\Pm_Y(E)\left(\alpha\frac{1}{ \Pm_X\Pm_Y(E)}\right)^{1/\alpha}\cdot\mathbb{E}_{\Pm_X\Pm_Y}^{1/\alpha}\left[\left(\frac{d\Pm_{XY}}{d\Pm_X\Pm_Y}\right)^{\alpha}\right] \left(\frac1\alpha \right)^{1/\alpha}\\
    &= \Pm_X\Pm_Y(E)^{\frac{\alpha-1}{\alpha}}\cdot \exp\left(\frac{\alpha-1}{\alpha} D_\alpha(\Pm_{XY}\|\Pm_X\Pm_Y\right).
\end{align}
\end{proof}
\ifCLASSOPTIONcaptionsoff
  \newpage
\fi



%
\bibliographystyle{IEEEtran}
\bibliography{bibtex/bib/sample}

\begin{thebibliography}{10}
\providecommand{\url}[1]{#1}
\csname url@samestyle\endcsname
\providecommand{\newblock}{\relax}
\providecommand{\bibinfo}[2]{#2}
\providecommand{\BIBentrySTDinterwordspacing}{\spaceskip=0pt\relax}
\providecommand{\BIBentryALTinterwordstretchfactor}{4}
\providecommand{\BIBentryALTinterwordspacing}{\spaceskip=\fontdimen2\font plus
\BIBentryALTinterwordstretchfactor\fontdimen3\font minus
  \fontdimen4\font\relax}
\providecommand{\BIBforeignlanguage}[2]{{%
\expandafter\ifx\csname l@#1\endcsname\relax
\typeout{** WARNING: IEEEtran.bst: No hyphenation pattern has been}%
\typeout{** loaded for the language `#1'. Using the pattern for}%
\typeout{** the default language instead.}%
\else
\language=\csname l@#1\endcsname
\fi
#2}}
\providecommand{\BIBdecl}{\relax}
\BIBdecl

\bibitem{genAdap}
C.~Dwork, V.~Feldman, M.~Hardt, T.~Pitassi, O.~Reingold, and A.~Roth,
  ``Generalization in adaptive data analysis and holdout reuse,'' in
  \emph{Proceedings of the 28th International Conference on Neural Information
  Processing Systems - Volume 2}.\hskip 1em plus 0.5em minus 0.4em\relax
  Cambridge, MA, USA: MIT Pressf, 2015.

\bibitem{maxInfo}
R.~M. Rogers, A.~Roth, A.~D. Smith, and O.~D. Thakkar, ``Max-information,
  differential privacy, and post-selection hypothesis testing,'' \emph{2016
  IEEE 57th Annual Symposium on Foundations of Computer Science (FOCS)}, pp.
  487--494, 2016.

\bibitem{learningMI}
R.~Bassily, S.~Moran, I.~Nachum, J.~Shafer, and A.~Yehudayoff, ``Learners that
  use little information,'' in \emph{Proceedings of Algorithmic Learning
  Theory}, ser. Proceedings of Machine Learning Research, F.~Janoos, M.~Mohri,
  and K.~Sridharan, Eds., vol.~83.\hskip 1em plus 0.5em minus 0.4em\relax PMLR,
  07--09 Apr 2018, pp. 25--55.

\bibitem{infoThGenAn}
A.~{Xu} and M.~{Raginsky}, ``{Information-theoretic analysis of generalization
  capability of learning algorithms},'' in \emph{Advances in Neural Information
  Processing Systems}, 2017, p. 2521–2530.

\bibitem{explBiasMI}
D.~Russo and J.~Zou, ``Controlling bias in adaptive data analysis using
  information theory,'' in \emph{Proceedings of the 19th International
  Conference on Artificial Intelligence and Statistics}, ser. Proceedings of
  Machine Learning Research, A.~Gretton and C.~C. Robert, Eds., vol.~51.\hskip
  1em plus 0.5em minus 0.4em\relax Cadiz, Spain: PMLR, 09--11 May 2016, pp.
  1232--1240.

\bibitem{leakageLong}
I.~{Issa}, A.~B. {Wagner}, and S.~{Kamath}, ``{An Operational Approach to
  Information Leakage},'' \emph{ArXiv e-prints}, jul 2018.

\bibitem{CompSecQuantLeakage}
C.~Braun, K.~Chatzikokolakis, and C.~Palamidessi, ``Quantitative notions of
  leakage for one-try attacks,'' \emph{Electronic Notes in Theoretical Computer
  Science}, vol. 249, pp. 75--91, 2009.

\bibitem{leakage}
I.~Issa, S.~Kamath, and A.~B. Wagner, ``An operational measure of information
  leakage,'' in \emph{2016 Annual Conference on Information Science and Systems
  (CISS)}, March 2016, pp. 234--239.

\bibitem{jiao2017dependence}
J.~Jiao, Y.~Han, and T.~Weissman, ``Dependence measures bounding the
  exploration bias for general measurements,'' in \emph{Information Theory
  (ISIT), 2017 IEEE International Symposium on}.\hskip 1em plus 0.5em minus
  0.4em\relax IEEE, 2017, pp. 1475--1479.

\bibitem{ISIT2018}
I.~Issa and M.~Gastpar, ``Computable bounds on the exploration bias,'' in
  \emph{2018 {IEEE} International Symposium on Information Theory, {ISIT} Vail,
  CO, USA, June 17-22, 2018}, 2018, pp. 576--580.

\bibitem{chainingMI}
\BIBentryALTinterwordspacing
A.~R. Asadi, E.~Abbe, and S.~Verd{\'{u}}, ``Chaining mutual information and
  tightening generalization bounds,'' \emph{CoRR}, vol. abs/1806.03803, 2018.
  [Online]. Available: \url{http://arxiv.org/abs/1806.03803}
\BIBentrySTDinterwordspacing

\bibitem{tighteningMI}
\BIBentryALTinterwordspacing
Y.~Bu, S.~Zou, and V.~V. Veeravalli, ``Tightening mutual information based
  bounds on generalization error,'' \emph{CoRR}, vol. abs/1901.04609, 2019.
  [Online]. Available: \url{http://arxiv.org/abs/1901.04609}
\BIBentrySTDinterwordspacing

\bibitem{genErrMISGLD}
A.~Pensia, V.~Jog, and P.-L. Loh, ``Generalization error bounds for noisy,
  iterative algorithms,'' \emph{2018 IEEE International Symposium on
  Information Theory (ISIT)}, pp. 546--550, 2018.

\bibitem{genErrWassDist}
A.~T. Lopez and V.~S. Jog, ``Generalization error bounds using wasserstein
  distances,'' \emph{2018 IEEE Information Theory Workshop (ITW)}, pp. 1--5,
  2018.

\bibitem{genErrWassDist2}
H.~Wang, M.~Diaz, J.~C. S.~S. Filho, and F.~P. Calmon, ``An
  information-theoretic view of generalization via wasserstein distance,''
  \emph{2019 IEEE International Symposium on Information Theory (ISIT)}, pp.
  577--581, 2019.

\bibitem{verduAlpha}
S.~Verd{\'{u}}, ``{\(\alpha\)}-mutual information,'' in \emph{2015 Information
  Theory and Applications Workshop, {ITA} 2015, San Diego, CA, USA, February
  1-6, 2015}, 2015, pp. 1--6.

\bibitem{alphaDiv}
T.~{van Erven} and P.~{Harremos}, ``Rényi divergence and kullback-leibler
  divergence,'' \emph{IEEE Transactions on Information Theory}, vol.~60, no.~7,
  pp. 3797--3820, July 2014.

\bibitem{opMeanRDiv1}
I.~{Csiszar}, ``Generalized cutoff rates and renyi's information measures,''
  \emph{IEEE Transactions on Information Theory}, vol.~41, no.~1, pp. 26--34,
  Jan 1995.

\bibitem{opMeanRDiv2}
P.~D. Gr\"{u}nwald, \emph{The Minimum Description Length Principle (Adaptive
  Computation and Machine Learning)}.\hskip 1em plus 0.5em minus 0.4em\relax
  The MIT Press, 2007.

\bibitem{infoRadius}
\BIBentryALTinterwordspacing
R.~Sibson, ``Information radius,'' \emph{Zeitschrift f{\"u}r
  Wahrscheinlichkeitstheorie und Verwandte Gebiete}, vol.~14, no.~2, pp.
  149--160, Jun 1969. [Online]. Available:
  \url{https://doi.org/10.1007/BF00537520}
\BIBentrySTDinterwordspacing

\bibitem{fDiv1}
\BIBentryALTinterwordspacing
F.~Liese and I.~Vajda, ``On divergences and informations in statistics and
  information theory,'' \emph{IEEE Trans. Inf. Theor.}, vol.~52, no.~10, pp.
  4394--4412, 2006. [Online]. Available:
  \url{http://dx.doi.org/10.1109/TIT.2006.881731}
\BIBentrySTDinterwordspacing

\bibitem{convergenceDiv}
A.~L. Gibbs and F.~E. Su, ``On choosing and bounding probability metrics,''
  \emph{Internat. Statist. Rev.}, pp. 419--435, 2002.

\bibitem{Su1995}
F.~E. Su, \emph{Methods for Quantifying Rates of Convergence for Random Walks
  on Groups, PhD Thesis}.\hskip 1em plus 0.5em minus 0.4em\relax Harvard
  University, 1995.

\bibitem{OrliczAmemiyaNorm}
\BIBentryALTinterwordspacing
H.~Hudzik and L.~Maligranda, ``Amemiya norm equals orlicz norm in general,''
  \emph{Indagationes Mathematicae}, vol.~11, no.~4, pp. 573 -- 585, 2000.
  [Online]. Available:
  \url{http://www.sciencedirect.com/science/article/pii/S0019357700800269}
\BIBentrySTDinterwordspacing

\bibitem{BLM2013Concentration}
S.~Boucheron, G.~Lugosi, and P.~Massart, \emph{Concentration Inequalities: A
  Nonasymptotic Theory of Independence}.\hskip 1em plus 0.5em minus 0.4em\relax
  Oxford University Press, 2013.

\bibitem{learningBook}
S.~Shalev-Shwartz and S.~Ben-David., \emph{Understanding machine learning: From
  theory to algorithms}.\hskip 1em plus 0.5em minus 0.4em\relax Cambridge
  University Press, 2014.

\bibitem{RenyiKLDiv}
T.~van Erven and P.~Harremos, ``R{\'e}nyi divergence and {K}ullback-{L}eibler
  divergence,'' \emph{IEEE Trans. Inf. Theory}, vol.~60, no.~7, pp. 3797--3820,
  July 2014.

\bibitem{fDivIneq}
\BIBentryALTinterwordspacing
I.~Sason and S.~Verdu, ``$f$ -divergence inequalities,'' \emph{IEEE Trans. Inf.
  Theor.}, vol.~62, no.~11, p. 5973–6006, Nov. 2016. [Online]. Available:
  \url{https://doi.org/10.1109/TIT.2016.2603151}
\BIBentrySTDinterwordspacing

\bibitem{statValidity}
C.~Dwork, V.~Feldman, M.~Hardt, T.~Pitassi, O.~Reingold, and A.~L. Roth,
  ``Preserving statistical validity in adaptive data analysis,'' in
  \emph{Proceedings of the Forty-seventh Annual ACM Symposium on Theory of
  Computing}.\hskip 1em plus 0.5em minus 0.4em\relax New York, NY, USA: ACM,
  2015, pp. 117--126.

\bibitem{bousqet}
\BIBentryALTinterwordspacing
O.~Bousquet and A.~Elisseeff, ``Stability and generalization,'' \emph{J. Mach.
  Learn. Res.}, vol.~2, pp. 499--526, 3 2002. [Online]. Available:
  \url{https://doi.org/10.1162/153244302760200704}
\BIBentrySTDinterwordspacing

\bibitem{subWeibull}
M.~Vladimirova and J.~Arbel, ``Sub-weibull distributions: generalizing
  sub-gaussian and sub-exponential properties to heavier-tailed
  distributions,'' 05 2019.

\bibitem{varReprfDiv}
\BIBentryALTinterwordspacing
X.~Nguyen, M.~J. Wainwright, and M.~I. Jordan, ``Estimating divergence
  functionals and the likelihood ratio by penalized convex risk minimization,''
  in \emph{Advances in Neural Information Processing Systems 20}, J.~C. Platt,
  D.~Koller, Y.~Singer, and S.~T. Roweis, Eds.\hskip 1em plus 0.5em minus
  0.4em\relax Curran Associates, Inc., 2008, pp. 1089--1096. [Online].
  Available:
  \url{http://papers.nips.cc/paper/3193-estimating-divergence-functionals-and-the-likelihood-ratio-by-penalized-convex-risk-minimization.pdf}
\BIBentrySTDinterwordspacing

\bibitem{regularCondProb}
\BIBentryALTinterwordspacing
J.~Pfanzagl, ``On the existence of regular conditional probabilities,''
  \emph{Zeitschrift f{\"u}r Wahrscheinlichkeitstheorie und Verwandte Gebiete},
  vol.~11, no.~3, pp. 244--256, Sep 1969. [Online]. Available:
  \url{https://doi.org/10.1007/BF00536383}
\BIBentrySTDinterwordspacing

\end{thebibliography}

%




\end{document}